\documentclass[a4paper]{article}
\usepackage[utf8]{inputenc}

\usepackage{amsmath,amssymb,amsthm}
\usepackage{lmodern} 
\usepackage{hyperref}
\usepackage{tikz}
\usepackage{xspace}
\usepackage{multirow}
\usepackage{booktabs}
\usepackage[ruled,vlined,linesnumbered]{algorithm2e}
\usepackage{cleveref} 
\usepackage{mathtools}
\usepackage{orcidlink}

\graphicspath{{figures/}}

\SetKwComment{tcc}{$\triangleright$~}{}
\SetKwComment{tcp}{$\triangleright$~}{}
\SetCommentSty{normalfont}

\newtheorem{theorem}{Theorem}
\crefname{theorem}{theorem}{theorems}
\Crefname{theorem}{Theorem}{Theorems}

\newtheorem{lemma}{Lemma}
\crefname{lemma}{lemma}{lemmas}
\Crefname{lemma}{Lemma}{Lemmas}

\newtheorem{corollary}{Corollary}
\crefname{corollary}{corollary}{corollaries}
\Crefname{corollary}{Corollary}{Corollaries}

\newtheorem{definition}{Definition}
\crefname{definition}{definition}{definitions}
\Crefname{definition}{Definition}{Definitions}

\newtheorem{problem}{Problem}
\crefname{problem}{problem}{problems}
\Crefname{problem}{Problem}{Problems}

\newtheorem{observation}{Observation}
\crefname{observation}{observation}{observations}
\Crefname{observation}{Observation}{Observations}

\crefname{proposition}{proposition}{propositions}
\Crefname{proposition}{Proposition}{Propositions}

\newtheorem{assumption}{Assumption}
\crefname{assumption}{assumption}{assumptions}
\Crefname{assumption}{Assumption}{Assumptions}

\newcommand{\ltdots}{..}
\newcommand{\gap}{\mathtt{-}}

\newcommand\rank{\mathrm{rank}}
\newcommand\select{\mathrm{select}}

\newcommand\msa{\textsf{MSA}\xspace}
\newcommand\msas{\textsf{MSA}s\xspace} 
\newcommand\msamn{\textsf{MSA}$[1\ltdots m,1\ltdots n]$\xspace}
\newcommand\msaij[2]{\mathsf{MSA}[#1,#2]\xspace}

\newcommand\efg{\textsf{EFG}\xspace}
\newcommand\efgs{\textsf{EFG}s\xspace}

\newcommand\spell{\mathsf{spell}\xspace}
\newcommand\sa{\mathsf{SA}\xspace}
\newcommand\mbwt{\mathsf{BWT}\xspace}
\newcommand\bwt{\textsf{BWT}\xspace}

\newcommand\gst{$\mathsf{GST}_\mathsf{MSA}$\xspace}
\newcommand\mgst{\mathsf{GST}_\mathsf{MSA}\xspace}
\newcommand\gpt{$\mathsf{GPT}_\mathsf{MSA}$\xspace}
\newcommand\mgpt{\mathsf{GPT}_\mathsf{MSA}\xspace}

\newcommand\atrue{\mathrm{\mathbf{true}}\xspace}
\newcommand\afalse{\mathrm{\mathbf{false}}\xspace}
\newcommand{\rroot}{\mathrm{root}}

\DeclareMathOperator{\pos}{pos}
\DeclareMathOperator{\cchar}{char}
\DeclareMathOperator{\sstring}{string}
\DeclareMathOperator{\parent}{parent}
\DeclareMathOperator{\stringdepth}{stringdepth}

\DeclareMathOperator{\rrank}{rank}
\newcommand{\pah}{\overline{H}}
\newcommand{\leftmostleaf}{\mathrm{leftmostleaf}}
\newcommand{\rightmostleaf}{\mathrm{rightmostleaf}}
\newcommand{\nextleaf}{\mathrm{nextleaf}}
\newcommand{\prevleaf}{\mathrm{prevleaf}}
\newcommand{\suffixlink}{\mathrm{suffixlink}}

\newcommand{\0}{\mathbf{0}}

\DeclarePairedDelimiter\ceil{\lceil}{\rceil}

\DeclareMathOperator*{\argmin}{\arg\!\min}

\newcommand{\prefix}{\preceq}
\newcommand{\properprefix}{\prec}

\title{Elastic Founder Graphs Improved and Enhanced\footnote{This is an extension of IWOCA 2022 \cite{RM22a} and CPM 2022 \cite{RM22b} conference papers, and of some results from the PhD dissertation projects of Massimo Equi \cite{Equi22} and Tuukka Norri \cite{Norri2022FounderSegmentations}.}}
\date{
Department of Computer Science, University of Helsinki, Finland
}
\author{
Nicola Rizzo\,\orcidlink{0000-0002-2035-6309}\\\texttt{nicola.rizzo@helsinki.fi} \and Massimo Equi\\\texttt{massimo.equi@helsinki.fi} \and Tuukka Norri\,\orcidlink{0000-0002-8276-0585}\\\texttt{tuukka.norri@helsinki.fi} \and Veli M\"akinen\,\orcidlink{0000-0003-4454-1493}\\\texttt{veli.makinen@helsinki.fi}
}

\begin{document}
\maketitle

\begin{abstract}
    Indexing labeled graphs for pattern matching is a central challenge of pangenomics.
    Equi et al.\ (Algorithmica, 2022) developed the \emph{Elastic Founder Graph} (\efg) representing an \emph{alignment} of $m$ sequences of length $n$, drawn from alphabet $\Sigma$ plus the special gap character: the paths spell the original sequences or their \emph{recombination}.
    By enforcing the \emph{semi-repeat-free} property, the \efg admits a polynomial-space index for linear-time pattern matching, breaking through the conditional lower bounds on indexing labeled graphs (Equi et al., SOFSEM 2021).
    In this work we improve the space of the \efg index answering
    pattern matching queries in linear time, from linear in the length of all strings spelled by three consecutive node labels, to linear in the size of the edge labels. Then, we develop linear-time construction algorithms optimizing for different metrics:
    we improve the existing linearithmic construction algorithms to $O(mn)$, by solving the novel \emph{exclusive ancestor set} problem on trees; we propose, for the simplified gapless setting, an $O(mn)$-time solution minimizing the maximum block height, that we generalize by substituting block height with \emph{prefix-aware height}.
    Finally, to show the versatility of the framework, we develop a BWT-based \efg index and study how to encode and perform document listing queries on a set of paths of the graphs, reporting which paths present a given pattern as a substring.
    We propose the \efg framework as an improved and enhanced version of the framework for the gapless setting, along with construction methods that are valid in any setting concerned with the segmentation of aligned sequences.
\\[1ex]\textbf{Keywords} multiple sequence alignment, pattern matching, segmentation algorithm, suffix tree, document listing
\\[1ex]\textbf{Funding} This work is partially funded by the European Union’s Horizon 2020 research and innovation programme under the Marie Skłodowska-Curie grant agreement No 956229 (ALPACA), under the European Research Council (ERC) grant agreement No.~851093 (SAFEBIO), by the Helsinki Institute for Information Technology (HIIT), and by the Academy of Finland (grant 351149).
\end{abstract}

\section{Introduction}
Searching strings in a graph has become a central problem along with the development of high-throughput sequencing techniques. Namely, thousands of human genomes are now available, forming a so-called \emph{pangenome} of a species \cite{DBLP:journals/bib/Consortium18}. Such pangenome can be used to enhance various analysis tasks that have previously been conducted with a single reference genome \cite{MNSV09jcb,Sch09,SVM14,Garetal18,hisat2,graphtyper2,Norrietal21}.
The most popular representation for a pangenome is a graph, whose paths spell the input genomes \cite{DBLP:journals/bib/Consortium18}. The basic primitive required on such pangenome graphs is to be able to search occurrences of query strings (short reads) as subpaths of the graph. Unfortunately, even finding exact matches of a query string of length $q$ in a graph with $e$ edges cannot be done significantly faster than $O(qe)$ time, and no index built in polynomial time allows for subquadratic-time string matching unless the Orthogonal Vectors Hypothesis (OVH) is false \cite{EMTSOFSEM2021,EGMT19}.
Therefore, practical tools deploy various heuristics or use other pangenome representations as a basis.

Due to the difficulty of string search in general graphs, Equi et al.~\cite{Equietal22} studied graphs obtained from \emph{multiple sequence alignments} (\msas), where an \msamn is a matrix composed of $m$ aligned rows that are strings of length $n$, drawn from an alphabet $\Sigma$ plus a special gap symbol $\gap$.
As we describe in \Cref{sect:definitions}, any segmentation of an \msa naturally induces a graph consisting of labeled nodes, partitioned into blocks.
The edges connect consecutive blocks and the original sequences are spelled as paths of this graph, but the graph enables the spelling of new sequences that are a \emph{recombination} of them.
Such \emph{elastic founder graph} (\efg) is illustrated in Figure~\ref{fig:segmentation}.
The key observation is that if the resulting node labels do not appear as a prefix of any other path than those starting at the same block, then the so-called \emph{semi-repeat-free} property holds and there is an index structure for the graph---based on the strings spelled by paths of a certain length---computable in polynomial time and supporting fast pattern matching.
Equi et al.~\cite{Equietal22} also showed that such indexability property is required, as in general the OVH-based lower bound holds for \efgs derived from \msas.

Moreover, an optimal segmentation resulting in a semi-repeat-free \efg is also computable in polynomial time.
For the simplified gapless setting, where the semi-repeat-free property is equivalent to the stronger \emph{repeat-free} property, we have two linear-time construction algorithms:
M\"akinen et al.~\cite{MCENT20}
gave an $O(mn)$ time algorithm to construct an indexable \efg minimizing its maximum block length, given a gapless \msamn; Equi et al.~\cite{Equietal22} reached the same complexity for the construction of an optimal \efg maximizing the number of blocks (which is a non-trivial task under the semi-repeat-free requirement).
In the example of \Cref{fig:segmentation}, these measures are equal to $5$ and $3$, respectively, since there are $3$ blocks defined by segments of length at most $5$.
The study of a third optimization criterion, consisting in minimizing the maximum block height, was left as an open case. This measure is defined as the maximum number of distinct strings segmented into a single block, and in the example of \Cref{fig:segmentation} this measure is equal to $3$, as the blocks contain at most $3$ nodes.
Equi et al.~\cite{Equietal22} also extended the two algorithms to general \msas aligned with gaps, obtaining an $O(mn \log m)$-time preprocessing algorithm which allows the construction of semi-repeat-free \efgs maximizing the number of blocks and, alternatively, minimizing the maximum length of a block, in $O(n)$ and in $O(n \log\log n)$ time, respectively.
We recall these results in \Cref{sect:efgconstruction}, and we refer the reader to the aforementioned papers for the connections between \efgs, Elastic Degenerate Strings, and Wheeler Graphs.

\begin{figure}
    \centering
    \includegraphics{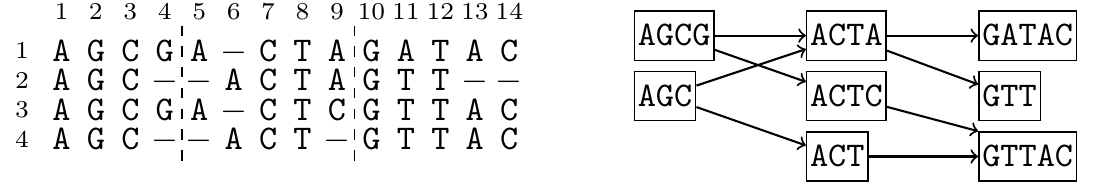}
    \caption{An elastic founder graph (right) induced from the segmentation of a multiple sequence alignment (left).
    The segmentation is visualized as a collection of vertical cuts (dashed lines).
    The \efg is \emph{semi-repeat-free}, that is, each node label occurs in the graph only as prefix of paths starting from the node's block. Thus, it is indexable for linear-time pattern matching.
    \label{fig:segmentation}}
\end{figure}

In this paper, we extend the theory on indexable \efgs as follows:
\begin{itemize}
    \item In \Cref{sect:index} we provide an index for \efgs capable of answering if a given pattern $Q$ appears as a subpath in the graph in time $O(\lvert Q \rvert)$, where the size of the index is linear in the space occupied by the concatenation of the edge labels. This keeps the same time complexity for queries while improving the space of the index by Equi et al.~\cite{Equietal22}, which is linear in the size of the strings spelled by paths involving three nodes.
    \item In \Cref{sect:efgconstruction} we study the construction of \efgs and we improve the $O(n \log \log n)$-time solution for minimizing the maximum segment length to $O(n)$ time, after the common preprocessing of the \msa for finding the valid segments.
    Moreover, we optimize for the \emph{height} of the resulting \efg.
    For the gapless case, we show how the left-to-right segmentation algorithm by Norri et al.\ exploiting \emph{left extensions}~\cite{NCKM19} can be combined with the computation of the \emph{minimal left extensions} by Equi et al.\ \cite{MCENT20} to obtain an $O(mn)$ algorithm in the case where $\lvert \Sigma \rvert \in O(m)$.
    Since left extensions cannot be used in the general and semi-repeat-free case, we develop an equivalent left-to-right solution exploiting \emph{meaningful right extensions} that is also correct in the case with gaps.
    However, the number of these extensions is $O(mn^2)$, and computing them takes $O(m n \alpha \log \lvert \Sigma \rvert)$ time, where $\alpha$ is the length of the longest run in the \msa where a row spells a prefix of the string spelled by another row, with $\alpha \in \Theta(n)$ in the worst case.
    Hence, we continue with a different generalization of the gapless height that we call \emph{prefix-aware height}, equal to the maximum number of distinct strings in a block but omitting strings that are prefixes of others strings of the block.
    In the example of \Cref{fig:segmentation}, this measure is equal to 2.
    \item In \Cref{sect:preprocessing} we study the preprocessing of the \msa for the segmentation algorithms.
    We improve the preprocessing algorithm finding the valid segments from $O(mn \log n)$ to $O(mn)$ time by performing an in-depth analysis of the existing solution based on the generalized suffix tree \gst of the gaps-removed \msa rows.
    Although removing gaps constitutes a loss of essential information, this information can be fed back into the structure by considering the right subsets of its nodes or leaves.
    Then, the main step in preprocessing the \msa is solving a novel ancestor problem on the tree structure of \gst that we call the \emph{exclusive ancestor set problem}, and as one of our contributions we identify this problem and provide a linear-time solution.
    The number of meaningful right extensions is $O(mn)$ for the prefix-aware height, and we obtain a $O(mn)$-time solution for computing them.
    The linear time is achieved thanks to the computation of the generalized suffix tree built from the \msa rows, its symmetrical prefix tree counterpart, and the constant-time navigation between the two offered by the suffix tree data structure of Belazzougui et al.\ answering weighted ancestor queries \cite{DBLP:conf/cpm/BelazzouguiKPR21}.
    \item Finally, in \Cref{sect:indexing-efg-with-paths} we develop a BWT-based variant of the \efg index enhanced with a set of paths $\mathcal{P}$ to solve a problem analogous to document listing queries \cite{Mut02}, namely, to report which paths present a given pattern as a substring. With this formulation one can, for example, restrict pattern matching along the rows of the original multiple sequence alignment, making the index functionally equivalent to those solving document listing on repetitive genome collections \cite{CN19,CMR20}.
\end{itemize}
In \Cref{sect:discussion}, we summarize and discuss our results.

\section{Definitions}\label{sect:definitions}
We follow the notation of Equi et al.~\cite{Equietal22}.

\paragraph{Strings}
We denote integer intervals by $[x..y]$. Let $\Sigma = [1..\sigma]$ be an alphabet of size $\lvert \Sigma \rvert = \sigma$.
A \emph{string} $T[1..n]$ is a sequence of symbols from $\Sigma$, in symbols $T\in \Sigma^n$, where 
$\Sigma^n$ denotes the set of strings of length $n$ over $\Sigma$.
In this paper, we assume that $\sigma$ is always smaller or equal to the length of the strings we are working with.
The \emph{reverse} of $T$, denoted with $T^{-1}$, is the string $T$ read from right to left.
We denote by $T[x..y]$ the \emph{substring} of $T$ made of the concatenation of its characters from the $x$-th to the $y$-th.
A \emph{suffix} (\emph{prefix}) of string $T[1..n]$ is $T[x..n]$ ($T[1..y]$) for $1\leq x \leq n$ ($1 \leq y \leq n$) and we say it is \emph{proper} if $x > 1$ ($y < n$).
If string $S$ is a prefix of string $T$, we write $S\prefix T$, and we write $S \properprefix T$ if $S$ is a proper prefix of $T$.
The \emph{length} of a string $T$ is denoted $|T|$ and the \emph{empty string} $\varepsilon$ is the string of length $0$.
In particular, substring $T[i..j]$ where $j<i$ is the empty string.
We denote with $\Sigma^*$ and $\Sigma^+$ the set of finite strings and finite non-empty strings over $\Sigma$, respectively.
String $Q$ \emph{occurs} in $T$ if $Q = T[x..y]$ for some interval $[x..y]$; in this case, we say that the starting position $x$ is an \emph{occurrence} of $Q$ in $T$, and $y = x + \lvert Q \rvert - 1$ is the \emph{ending position} of such occurrence.
The \emph{lexicographic order} of two strings $A$ and $B$ is naturally defined by the order of the alphabet: $A<B$ if and only if $A[1..y]=B[1..y]$ and $A[y+1]<B[y+1]$ for some $y \geq 0$.
If $y + 1 > \min(|A|,|B|)$, then the shorter one is regarded as smaller. However, we usually avoid this implicit comparison by adding an \emph{end marker} $\$$ to the strings, where $\$$ does not occur in any of the strings, and we consider $\$$ to be the smallest character lexicographically.
The concatenation of strings $A$ and $B$ is denoted as $A \cdot B$, or just $AB$.

\paragraph{Model of Computation} In the following, we assume the word RAM model of computation~\cite{DBLP:conf/stoc/FredmanW90}.

\paragraph{Elastic founder graphs}
\msas can be compactly represented by elastic founder graphs, the vertex-labeled graphs that we formalize in this paragraph.

A \emph{multiple sequence alignment} \msamn is a matrix whose rows are $m$ strings of length $n$ drawn from $\Sigma \cup \{\gap\}$.
Here, $\gap \notin \Sigma$ is the \emph{gap} symbol and we indicate $[1..m]$ and $[1..n]$ the set of row and column indices, respectively.
For a string $T \in \left(\Sigma \cup \{\gap\}\right)^*$, we denote with $\spell(T)$ the string resulting from removing the gap symbols from $T$.
If an \msa does not contain gaps then we say it is \emph{gapless}, otherwise, we say that it is a \emph{general} \msa.
Given $I \subseteq [1..m]$, we denote with $\msaij{I}{1..n}$ the \msa obtained by considering only rows $\msaij{i}{1..n}$ with $i \in I$.

Let $\mathcal{P}$ be a \emph{partitioning} of $[1..n]$, that is, a sequence of subintervals $\mathcal{P}=[x_1..y_1],[x_2..y_2],\dots,[x_b..y_b]$ where $x_1=1$, $y_b=n$, and for all $j>2$, $x_j=y_{j-1}+1$.
A \emph{segmentation} $S$ of \msamn based on partitioning $\mathcal{P}$ is the sequence of $b$ sets $S^k= \{\spell(\msaij{i}{x_k..y_k}) \mid 1\leq i\leq m\}$ for $1\leq k\leq b$;
in addition, we do not allow segments corresponding to a full run of gaps in the \msa.

\begin{assumption}\label{ass:empty}
To obtain a proper segmentation $S_1, \dots, S_b$ of \msamn we require that $\varepsilon \notin S_k$ for any $k \in [1..b]$.
\end{assumption}

We call set $S^k$ a \emph{block} and we refer to $\msaij{1..m}{x_k..y_k}$ or just $[x_k..y_k]$ as a \emph{segment} of $S$, since in the rest of paper $S$ will always be implicitly associated to some partitioning $\mathcal{P}$ and \msamn.
Then, the \emph{length} of a segment $[x..y]$ is simply $y - x + 1$.
The \emph{height} of a block $S^k$ is $H(S^k) \coloneqq \lvert S^k \rvert$, and
given a segment $[x..y]$ we denote the height of the corresponding block---its \emph{segment height}---with $H(\msaij{1..m}{x..y})$ or just $H([x..y])$.  

The segmentation of gapless \msas naturally leads to the definition of a founder graph through the block graph concept.
\begin{definition}[Block Graph]
A \emph{block graph} is a graph $G=(V,E,\ell)$ where $V$ is the set of nodes, $E \subseteq V \times V$ is the set of edges, and $\ell: V \rightarrow \Sigma^+$ is a function that assigns a non-empty string label to every node. In addition, the following properties must hold:
\begin{enumerate}
	\item set $V$ can be partitioned into a sequence of $b$ \emph{blocks} $V^1, V^2, \ldots, V^b$, that is, $V = V^1 \cup V^2 \cup \cdots \cup V^b$ and $V^i \cap V^j = \emptyset$ for all $i\neq j$;
	\item all edges connect consecutive blocks, that is, if $(v,w) \in E$ then $v \in V^i$ and $w \in V^{i+1}$ for some $1 \leq i \leq b-1$; and
	\item the labels of a block are of equal length, that is, if $v,w \in V^i$ then $|\ell(v)| = |\ell(w)|$ and if $v\neq w$ then $\ell(v) \neq \ell(w)$, for each $1 \leq i \leq b$.
\end{enumerate}
If there is only one block, meaning that $b = 1$ and thus $E = \emptyset$, we say that $G$ is \emph{trivial}.
\end{definition}
Block $S^k$ equals segment $\msaij{1..m}{x_k..y_k}$ and the \emph{founder graph} is a block graph induced by segmentation $S$~\cite{MCENT20}.
The idea is to have a graph in which the nodes represent the strings in $S$ while the edges retain the information of how such strings can be recombined to spell any sequence in the original \msa.
Alternatively, in \Cref{sect:indexing-efg-with-paths} we will describe an index data structure for querying the graph considering only a predefined set of paths.

With \emph{general} \msas, we consider the following generalization.
\begin{definition}[Elastic Block Graphs and Founder Graphs]\label{def:efg}
We call a block graph \emph{elastic} if its third condition is relaxed in the sense that each $V^i$ can contain non-empty variable-length strings. An \emph{elastic founder graph} (\efg) is an elastic block graph $G(S) = (V,E,\ell)$ \emph{induced} by a segmentation $S$ as follows: for each $1 \leq k \leq b$ we have $S^k = \{\spell(\msaij{i}{x_k..y_k}) \mid 1\leq i\leq m\} = \{\ell(v) : v \in V^k\}$.
It holds that $(v,w) \in E$ if and only if there exist block $k \in [1 \ltdots b-1]$ and row $i \in [1 \ltdots m]$ such that $v \in V^k$, $w \in V^{k+1}$, and $\spell(\msaij{i}{x_k..y_{k+1}})= \ell(v)\ell(w)$.
\end{definition}

For example, in the general $\msaij{1..4}{1..13}$ of \Cref{fig:segmentation}, the segmentation based on partitioning $[1..4],[5..9],[10..14]$ induces an \efg $G(S) = (V^1 \cup V^2 \cup V^3, E, \ell)$ where the nodes in $V^1$, $V^2$, and $V^3$ have labels of variable length.
As noted by Equi et al.\ \cite{Equietal22}, Block Graphs are connected to Generalized Degenerate Strings \cite{Alzetal20} and Elastic Founder Graphs are connected to Elastic Degenerate Strings \cite{bernardini_et_al2019elastic}.

By definition, (elastic) founder and block graphs are acyclic. For convention, we interpret the direction of the edges as going from left to right. 
Consider a \emph{path} $P = v_1 \cdots v_k$ in $G$ between any two nodes, where we define a path of length $k$ as a sequence of $k$ vertices connected by edges, that is, $(v_1,v_2),(v_2,v_3),\ldots,(v_{k-1},v_k) \in E$. 
The label $\ell(P) \coloneqq \ell(v_1) \cdots \ell(v_k)$ of $P$ is the concatenation of the labels of the nodes in the path.
Let $Q$ be a query string. We say that $Q$ \emph{occurs} in $G(S)$ if $Q$ is a substring of $\ell(P)$ for any path $P$ of $G(S)$.
In this case, for simplicity we say that $Q$ occurs in $G(S)$ as a substring of $P$, to mean a substring of $\ell(P)$.

Finally, we can introduce the key property making $\efgs$ indexable for pattern matching.
\begin{definition}[\cite{MCENT20}]\label{def:repeat-free}
\efg $G(S)$ is \emph{repeat-free} if each $\ell(v)$ for $v\in V$ occurs in $G(S)$ only as a prefix of paths starting with $v$.
\end{definition}
\begin{definition}[\cite{MCENT20}]\label{def:semi-repeat-free}
\efg $G(S)$ is \emph{semi-repeat-free} if each $\ell(v)$ for $v\in V$ occurs in $G(S)$ only as a prefix of paths starting with $w\in V$, where $w$ is from the same block as $v$. 
\end{definition}
For example, the \efg of \Cref{fig:segmentation} is not repeat-free, since $\mathtt{AGC}$ occurs as a prefix of two distinct labels of nodes in the same block, but it is semi-repeat-free since all node labels $\ell(v)$ with $v \in V^k$ occur in $G(S)$ only starting from block $V^k$, or they do not occur at all elsewhere in the graph.
Note that in the gapless setting, or more generally if the \efg is non-elastic, the repeat-free and semi-repeat-free notions are equivalent.

\paragraph{Basic tools}

A \emph{trie} or \emph{keyword tree}~\cite{Bri59} of a set of strings is a rooted directed tree where the outgoing edges of each node are labeled by distinct symbols and there is a unique root-to-leaf path spelling each string in the set; the shared part of two root-to-leaf paths spells the longest common prefix of the corresponding strings.
In a \emph{compact trie}, the maximal non-branching paths of a trie become edges labeled with the concatenation of labels on the path.
The \emph{suffix tree} of $T \in \Sigma^*$ is the compact trie of all suffixes of string $T' = T \mathbf{0}$, with $\mathbf{0} \notin \Sigma$ a terminator character.
In this case, the edge labels are substrings of $T$ and can be represented in constant space as an interval. Such a tree takes linear space and can be constructed in linear time, assuming that $\sigma \le \lvert T \rvert$, so that when reading the leaves from left to right the suffixes are listed in their lexicographic order~\cite{DBLP:journals/algorithmica/Ukkonen95,Farach97}.
We say that two or more leaves of the suffix tree are \emph{adjacent} if they succeed one another in lexicographic order. The leaves form the \emph{suffix array} $\sa_{T'}[1..|T'|]$, where $\sa_{T'}[i]=j$ iff $T'[j..|T'|]$ is the $i$-th smallest suffix in lexicographic order~\cite{MM93}.
A \emph{generalized suffix tree} or \emph{array} is one built on a set of $m$ strings~\cite{DBLP:books/cu/Gusfield1997}.
In this case, string $T$ above is the concatenation of the strings after appending a unique end marker $\$_i$ to each string\footnote{For our purposes, the suffix tree of the concatenated strings is functionally equivalent to the ``trimmed'' generalized suffix tree seen in \Cref{fig:gstmsa}. Also, we can use just one unique terminator $\mathbf{0}$.
}, with $1 \le i \le m$. When applied to string $T$, the \emph{Burrows--Wheeler transform} \cite{BW94} yields another string $\mbwt_T$ such that $\mbwt_T[i] = T'[\sa_{T'}[i] - 1]$ where $T'$ wraps, that is, $T'[-1] = T'[|T| + 1] = \$$.

Let $Q[1..m]$ be a query string. If $Q$ occurs in $T$, then the \emph{locus} or \emph{implicit node} of $Q$ in the suffix tree of $T$ is $(v,k)$ such that $Q = XY$, where $X$ is the path spelled from the root to the parent of $v$ and $Y$ is the prefix of length $k$ of the edge from the parent of $v$ to $v$. The leaves in the subtree rooted at $v$, or \emph{the leaves covered by $v$}, are then all the suffixes sharing the common prefix $Q$.
Let $aX$ and $X$ be the paths spelled from the root of a suffix tree to nodes $v$ and $w$, respectively. Then one can store a \emph{suffix link} from $v$ to $w$.
For suffix trees, a \emph{weighted ancestor query} asks for the computation of the implicit or explicit node corresponding to substring $T[x..y]$ of the text, given $x$ and $y$.

String $B[1..n]$ from a binary alphabet is called a \emph{bitvector}.
Let $\rrank_1(B,i)$ or just $\rrank(B,i)$ be the number of 1s in $B[1..i]$, and analogously $\rrank_0(B,i)$ returns the number of $0$s in $B[1..i]$.
Operation $\mathrm{select}(B,j)$ returns the index $i$ containing the $j$-th 1 in $B$.
Both queries can be answered in constant time using an index constructible in linear time and requiring $o(n)$ bits of space in addition to the bitvector itself~\cite{Jac89}.

\section{Indexing semi-repeat-free \efgs}\label{sect:index}
Our goal is to preprocess an \efg $G$ so that we can check in $O(|Q|)$ time if a query string $Q$ occurs in $G$, in the same fashion as in the following result by Equi et al.
\begin{theorem}[{\cite[Theorem 8]{Equietal22}}]
    A semi-repeat-free \efg $G = (V,E,\ell)$ can be indexed in polynomial time into a data structure occupying $O(\lvert D \rvert \log \lvert D\rvert)$ bits of space, where $\lvert D \rvert = (N H^2)$, $N$ is the total length of the node labels, and $H$ is the height of $G$.
    Later, one can find out in $O(\lvert Q \rvert)$ time if a given query string $Q$ occurs in $G$.
\end{theorem}

This result is based on string $D = \sum_{(u,v), (v,w) \in E} ((\ell(u)\ell(v)\ell(w))^{-1} \mathbf{0}$, whereas we now show that it is sufficient to build an index based on the edges of $G$.
We first concatenate edge labels $\ell(v)\ell(w)$, for each $(v,w) \in E$, in strings
\[
    D_F = \prod_{(v,w)\in E} (\ell(v)\ell(w))\mathbf{0} \qquad\text{and}\qquad
    D_R = \prod_{(v,w)\in E} (\ell(v)\ell(w))^{-1}\mathbf{0}\$^{|\ell(v)|}\mathbf{0}
\]
where $\mathbf{0}, \$ \notin \Sigma$ and $\$^{|\ell(v)|}$ is the dollar sign repeated as many times as there are characters in string $\ell(v)$. The reason for this unary encoding of the length of the string will be clear later. We now construct the suffix tree of both $D_F$ and $D_R$.
The matching algorithm considers three possible cases: a match spans at most two nodes, exactly three nodes, or at least four nodes. The next case is considered when the previous one can be ruled out.

\subsection{A match spans at most two nodes} 
We run a standard pattern matching query for $Q^{-1}$ in the suffix tree of $D_R$ which, if successful, locates a match for query string $Q$ spanning at most two nodes. If the query fails, that is, $Q^{-1}$ is not a substring of $D_R$, we know that if a match exists then it has to span at least three nodes.

To locate these other potential matches, first, we want to know the number $k$ of characters we matched in the query before we reached a block boundary:
\begin{itemize}
    \item if the query fails because of a mismatch and there is a branch with a $\mathbf{0}$, then $k$ is the length of the prefix of $Q^{-1}$ matched so far;
    \item if the query fails because of a mismatch and there is no branch with a $\mathbf{0}$, then we backtrack on the path that we matched so far to the last internal node from which there was a downward path starting with a $\mathbf{0}$; in this case, $k$ is the length of the prefix of $Q^{-1}$ matched until this internal node.
\end{itemize}
In other words, $k$ is the length of the longest suffix of $Q$ matching a subpath of the \efg that spans at most two nodes and starts from the start of a node label (i.e.\ the block boundary).

In the event that there is no such node, meaning that no node label starts with the last character of $Q$, then we know that there is no match.
Moreover, we have no match also if we have matched no full node label so far.
To understand if we are in this situation, at construction time we can mark all the internal nodes reached by a path that contains at least a full node label.
Alternatively, starting from the node to which we backtracked in the suffix tree of $D_R$, we can also check if we can read any string $\mathbf{0}\mathbf{\$}^s\mathbf{0}$ for $1 \le s \le k$.
If this test is negative, by construction of $D_R$ we did not encounter a full node label and $Q$ does not occur in $G$.
In \Cref{sect:indexing-efg-with-paths} we will propose another method, based on suffix arrays, to identify if we have read a full node label, by focusing on node labels $\ell(v)$ that do not contain other node labels as their prefix.

Let $j=|Q|-k$.
So far we have found a match for $Q^{-1}[1..k]$ in $D_R$ spanning at least one full node label.
The semi-repeat-free property constrains the occurrences of $Q[j+1..\lvert Q \rvert]$ to start from some block $V^{r-1}$.
This string can span multiple edges between blocks $V^{r-1}$ and $V^{r}$, and can even appear as a prefix of the node labels in $V^{r-1}$, as Figure~\ref{fig:MatchEnd} and the next lemma illustrate.
\begin{figure}
    \centering
    \includegraphics{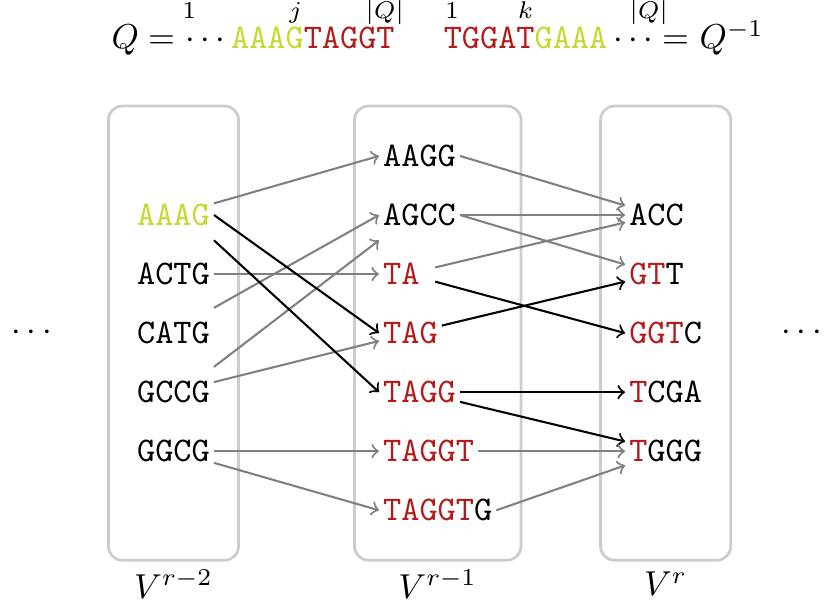}
    \caption{The last three blocks in which a suffix of string $Q$ matches. We highlight the last $k = \lvert Q \rvert - j$ characters and their matches in red, while the lime characters mark the unique match in the second to last block.}
    \label{fig:MatchEnd}
\end{figure}
\begin{lemma}\label{lemma:blocksync}
Let $Q$ be a string such that there exists at least one suffix of $Q$ that matches in semi-repeat-free EFG $G=(V,E,\ell)$ spanning at least three blocks.
Let $V^{r}$ be the rightmost block containing matched characters in some occurrence of $Q$ in $G$, and consider suffix $S=Q[|Q|-j+1..|Q|]$ that matches a prefix of $\ell(v)\ell(w)$ for some $v \in V^{r-1}$, $w \in V^r$, for $1\leq k \leq |Q|$ and $2\leq j\leq |Q|$.
Suffix $S$ can appear only as prefix of $\ell(v')\ell(w')$, for some $v' \in V^{r-1}$, $w' \in V^r$.
\end{lemma}
\begin{proof}
Consider string $R=\ell(v')\ell(w')$ for some $v' \in V^{r-1}$, $w' \in V^r$, and assume by contradiction that $S=R[k_1...k_2]$, for some $1<k_1<k_2\leq|R|$; that is $S$ matches not as a prefix of $\ell(v')\ell(w')$. Then, label $\ell(v)$ appears as a prefix of $R[k_1...k_2]$ and not as a prefix of $R$, contradicting the assumption of $G$ being semi-repeat-free.
\end{proof}

\subsection{A match spans exactly three nodes}
We have matched $Q^{-1}[1..k]$ in $D_R$ and we have located the corresponding left boundary of $V^{r-1}$ which, thanks to Lemma~\ref{lemma:blocksync}, we know to be a unique location in the EFG.
For now, let us consider only the occurrences of $Q^{-1}[1..k]$ that match fully in $V^{r-1}$ and partially in $V^r$, and let us mark every node in $V^{r-1}$ hit by this match, that is all $v_{\mathrm{R}} \in V^{r-1}$ such that $\ell(v)$ is a proper prefix of $Q[j+1..|Q|]$ and $v_{\mathrm{R}}$ has an outneighbor whose label partially reads the remaining suffix of $Q$.
We need to consider each and every one of these nodes, because each of them may or may not be connected to a node in $V^{r-2}$ spelling the rest of the pattern, as seen in \Cref{fig:MatchEnd}.

On one hand, notice that these nodes are all one prefix of another, and thus they are at most $|Q^{-1}[1..k]|-1=O(|Q|)$.
On the other hand, \efgs have few restrictions on their topology and we do not want to explore the in-neighborhoods of all of these nodes, as we would spend more than $O(\lvert Q \rvert)$ time.
Instead, we can exploit the number that we have encoded in unary after each edge in $D_R$.
Indeed, consider the path reading $\$$ characters from the node corresponding to $Q^{-1}[1..k] \mathbf{0}$ in $D_R$.
It is easy to see that there is an occurrence of $Q[j+1..\lvert Q \rvert]$ in the \efg decomposed as $\ell(v_{\mathrm{R}}) \cdot \ell(w)[1..k-\lvert\ell(v_\mathrm{R})\rvert]$, with $v_\mathrm{R} \in V^{r-1}$ and $w \in V^{r}$, if and only if $Q^{-1}[1..k] \mathbf{0} \$^{\lvert \ell(v_\mathrm{R}) \rvert} \mathbf{0}$ is a substring of $D_R$.
Thus, all the nodes can be marked by checking if $\$^p\mathbf{0}$ can be read after $Q^{-1}[1..k]\mathbf{0}$ in $D_R$, for $1 \le p \le \lvert Q \rvert$, and this can be done by simply descending down the aforementioned path.

Then, if we want to find a match of $Q$ spanning exactly three nodes, we search for $Q$ in the suffix tree of $D_F$, and we locate nodes $v_\mathrm{F} \in V^{r-1}$ such that there exists $u \in V^r$ with $(u,v_\mathrm{F}) \in E$, $\ell(u) = Q[1..j]$, and $\ell(v_\mathrm{F}) = Q[j+1..j + \lvert \ell(v_\mathrm{F}) \rvert]$.
Exploiting the information that we already know---that $Q[j+1..\lvert Q \rvert]$ starts from block $V^{r-1}$---this forward search is easier than the previous backward search: due to the semi-repeat-free property, there is a match of $Q$ in $G$ if and only if $Q[1..j] \ell(v_\mathrm{F}) \mathbf{0}$ occurs in $D_F$, with $v_\mathrm{F}$ equal to some $v_\mathrm{R}$ previously marked in the backward search.
We can easily preprocess the suffix tree of $D_F$ to locate nodes $v_\mathrm{F}$, since the semi-repeat-free property guarantees that any leaf in the subtree corresponding to string $T \ell(v) \mathbf{0}$ for $T \in \Sigma^*$ uniquely identifies any $v \in V$.
Thus, we query $Q[1..j] \cdot Q[j+1..j+1+s] \cdot \mathbf{0}$ for $0 \le s < k$, in a simple descent of the suffix tree of $D_F$, and we check whether the corresponding nodes $v_\mathrm{F}$ are marked.

Notice that the procedure described above works correctly even if $Q^{-1}[1..k]$ matches two full node labels.
Indeed, if a suffix of $Q$ is matching an entire node label in $V^{r-1}$, and there is an edge to that node from another fully matched node label in $V^{r-2}$, then we find such node label in $V^{r-2}$ with our first search of $Q^{-1}$ in $D_R$, and we can perform the procedure just described making $V^{r-1}$ play the role of $V^{r}$ and $V^{r-2}$ play the role of $V^{r-1}$.

\subsection{A match spans at least four nodes}
If we have not found a full match until now it means that, if there is any, it must span at least four nodes.
Consider the first search in $D_R$ during which we discovered position $j$. At the end of that search, we identified a specific internal node in the suffix tree of $D_R$ corresponding to $Q[j+1..|Q|]$. From that internal node, there is a downward path starting with $\mathbf{0}$ and continuing with a series of $\$$ characters.
From the previous case, we know to have marked nodes $v_\mathrm{R} \in V^{r-1}$ that are matching a prefix of $Q[j+1..|Q|]$, which are at most $O(|Q|)$.
Let $\hat{v}$ be the marked node with the shortest label length $\lvert \ell(\hat{v}) \rvert$.
We now search $Q^{-1}[k - \lvert \ell(\hat{v}) \rvert + 1..\lvert Q \rvert$] in $D_R$.
This way, we locate the only node $u$ such that $Q[j'..j]=\ell(u)$ in full, for some $1 \le j' \le j$, and $(u,v) \in E$ with $\ell(\hat{v})$ a prefix of $\ell(v)$.
Notice that we already ruled out the possibility of a match spanning three blocks, thus $Q[j'..j]$ must match $\ell(u)$ in full, not just one of its suffixes.
Moreover, notice that if $Q[j+1..|Q|]$ fully matched a prefix of a node in $V^{r-1}$ from which there was an edge to some $v_\mathrm{R}$, then we would have found this match with the very first search that we run.
Since this is not the case, we need to do the following.
We scan each marked node $v_\mathrm{R} \in V^{r-1}$, and we check whether edge $(u,v)$ exists: if this is the case then suffix $Q[j'+1..|Q|]$ has a match spanning blocks $V^{r-2}$, $V^{r-1}$, and $V^r$.
To perform this scan, we search $\ell(u) \cdot Q[j+1..j + 1 + s] \cdot \mathbf{0}$ for $0 \le s < k$ in $D_F$ in the same way as in nodes $v_\mathrm{F}$ of last paragraph, and every time we find a match we check if the corresponding node in $V^{r-1}$ is marked.\footnote{Note that we unambiguously identified node $u$ as part of any remaining match of $Q$ in $G$.
The original solution by Equi et al.\ computes the tries corresponding to the labels of the outneighborhood of each node $v$, and also the tries of the (reversed) strings of each block, which we could employ here for navigation.
Nonetheless, $D_R$ and $D_F$ can also be used to navigate $G$ in a similar way.}

At this point, from this block to the left, we can continue a match in the previous block using the suffix tree of $D_R$ as before, always using the label of the current node concatenated with the prefix of $Q$ that is yet to be matched.
This works because each block to the left, except the last one, has only one node label matching in $|Q|$, as the following lemma shows.
\begin{lemma}[{\cite[Lemma 9]{Equietal22}}]\label{lemma:keyproperty}
\label{lemma:uniquev}
Consider a semi-repeat free \efg $G=(V,E)$. 
String $\ell(v)\ell(w)$, where $(v,w)\in E$, can only appear in $G$ as a prefix of paths starting with $v$.
\end{lemma}
\begin{proof}
Assume for contradiction that $\ell(v)\ell(w)$ is a prefix of a path starting inside the label of some $v'\in V$, $v'\neq v$. Then $\ell(v)$ is a prefix of such path, and this is only possible if $v'$ is in the same block as $v$ and $\ell(v)$ is a proper prefix of $\ell(v')$: otherwise, $G$ would not be semi-repeat free. Then $|\ell(v)|<|\ell(v')|$ and $\ell(w)$ has an occurrence in a path starting inside the label $\ell(v')$. This is a contradiction of the fact that $G$ is semi-repeat-free. 
\end{proof}
Thus, once we identified position $j$, by \Cref{lemma:uniquev} all the paths having $Q[j'..|Q|]$ as a prefix must start with $v$, because $Q[j'..j]=\ell(v)$, and the same arguments hold for all edges that are fully spelled in $Q$.
This backward search in $D_R$ is able to complete the match also in the last block, possibly matching all the characters in $Q$ before reaching a leaf of $D_R$.

Thus, after the linear-time construction of the suffix tree of $D_R$, each position of $\lvert Q \rvert$ is considered at most three times by our searches, so the matching of $Q$ in $G$ can be done in time $O(\lvert Q \rvert)$.
Also, note that the nodes marked in block $V^{r-1}$ are uniquely identified by a range of $Q$, both in the backward $D_R$ search and in the forward $D_F$ one, so we can use the suffix trees as black-box indices for matching a pattern in $D_F$, $D_R$, and instead mark positions of $Q$ in a bit-vector of size $\lvert Q \rvert$ for that part of the algorithm.
\begin{theorem}\label{theo:index}
    A semi-repeat-free \efg $G = (V,E,\ell)$ can be indexed in polynomial time into a data structure occupying $O(\lvert D \rvert \log \lvert D\rvert)$ bits of space, where $\lvert D \rvert = O(N H)$, $N$ is the total length of the node labels, and $H$ is the height of $G$.
    Later, one can find out in $O(\lvert Q \rvert)$ time if a given query string $Q$ occurs in $G$.

\end{theorem}

\section{Construction of optimal \efgs}\label{sect:efgconstruction}
In this section, we review and expand the theory of EFG construction algorithms.
Recall that in the absence of gaps the semi-repeat-free and repeat-free notions (\Cref{def:repeat-free,def:semi-repeat-free}) are equivalent since the strings of any block cannot have variable length.
Indeed, after showing that semi-repeat-free \efgs are easy to index for fast pattern matching, Equi et al.\ \cite{Equietal22} extended the previous results for the gapless setting showing that semi-repeat-free \efgs are equivalent to specific segmentations of the \msa: the semi-repeat-free property has to be checked only against the \msa, and not the final \efg.
We recall these arguments in \Cref{sub:segmentation}, along with the resulting recurrence to compute an optimal segmentation under three score functions:
$i$.\ maximizing the number of blocks;
$ii$.\ minimizing the maximum length of a segment;
and $iii$.\ minimizing the maximum height of a block.

In the gapless and repeat-free setting, scores $i$.\ and $ii$.\ admit the construction of indexable founder graphs in $O(mn)$ time, thanks to previous research on founder graphs and \msa segmentations~\cite{MCENT20,NCKM19,CKMN19}.
In \Cref{sub:gaplesssolution} we combine these results to obtain an $O(mn)$ time solution for score $iii$.\ as well: the optimal segmentation is found by first computing the \emph{meaningful left extensions}, that is, the positions $x_1 > \dots > x_k$ where the height of repeat-free segment $[x_i..y]$ increases, with $y \in [1..n]$.

In the general and semi-repeat-free setting, extending a segment to the left can violate the semi-repeat-free property and the height can decrease.
Thus, Equi et al.\ in \cite{Equietal22} gave $O(n)$- and $O(n \log \log n)$-time algorithms for scores $i$.\ and $ii$., respectively, exploiting the semi-repeat-free \emph{right extensions} after a common $O(mn \log m)$-time preprocessing of the \msa, that we review in \Cref{sub:efgconstruction}.
In \Cref{sect:minmaxlength} we improve the construction algorithm for score $ii.$ to $O(n)$ time and in \Cref{sect:preprocessing} we improve the preprocessing to $O(mn)$, reaching global linear time.

In \Cref{sub:mainalgorithm}, we develop a similar algorithm for the construction of a semi-repeat-free segmentation that is optimal for score $iii.$, processing the \emph{meaningful right extensions}.
Although the number of these extensions is $O(n^2)$ in total, we manage to provide a parameterized linear-time solution providing an upper bound based on the length of the longest run where any two rows spell strings that are one prefix of the other.
Instead, an alternative notion of height, the \emph{prefix-aware height}, generates $O(mn)$ \emph{meaningful prefix-aware right extensions}: they can be processed in the same fashion as the original height to obtain an optimal segmentation, and we will show how to compute them efficiently in \Cref{sect:preprocessing}.

\subsection{Segmentation characterization for indexable \efgs}\label{sub:segmentation}
We will discuss the repeat-free and semi-repeat-free properties (\Cref{def:repeat-free,def:semi-repeat-free}) together as the \mbox{(semi-)}repeat-free property, when applicable.
However, this parenthesis notion is not always applicable, as we focus on \msas with gaps and most of our results hold only with the semi-repeat-free property.

Consider a segmentation $S = S^1, S^2, \ldots, S^b$ that induces a (semi-)repeat-free \efg $G(S)=(V,E,\ell)$, as per \Cref{def:efg}. The strings occurring in graph $G(S)$ are a superset of the strings occurring in the original \msa rows because each node label can represent \emph{multiple} rows and each edge $(v,w) \in E$ means the existence of \emph{some} row spelling $\ell(v)\ell(w)$ in the corresponding consecutive segments. For example, string $\mathtt{GACTAGT}$ occurs in the \efg of \Cref{fig:segmentation} but it does not occur in any row of the original \msa.

The (semi-)repeat-free property involves graph $G(S)$, but luckily it does not depend on these new strings added in the founder graph and can be checked only against the \msa and segmentation $S$. Intuitively, this is because the added strings involve three or more vertices of $G(S)$. This simplifies choosing a segmentation resulting in an indexable founder graph and it was initially proven by Mäkinen et al.\ in the gapless and repeat-free setting.

\begin{lemma}[Characterization, gapless setting~\cite{MCENT20}]
We say that a segment $[x..y]$ of a \emph{gapless} $\msaij{1..m}{1..n}$ is repeat-free if string $\msaij{i}{x..y}$ occurs in the \msa only at position $x$ of some row, for all $1 \le i \le m$. Then $G(S)$ is repeat-free if and only if all segments of $S$ are repeat-free.
\end{lemma}
Equi et al.\ in \cite{Equietal22} refined this property for \msas with gaps, but did not provide an explicit proof. Since it is essential for the correctness of the construction algorithms, we provide such a proof here.

\begin{lemma}[Characterization~\cite{Equietal22}]\label{lem:characterization}
We say that segment $[x..y]$ of a \emph{general} $\msaij{1..m}{1..n}$ is semi-repeat-free if for any $i,i' \!\in\! [1..m]$ string $\spell(\msaij{i}{x..y})$ occurs in gaps-removed row $\spell(\msaij{i'}{1..n})$ only at position $g(i',x)$, where $g(i',x)$ is equal to $x$ minus the number of gaps in $\msaij{i'}{1..x}$. Similarly, $[x..y]$ is repeat-free if the possible occurrence of $\spell(\msaij{i}{x..y})$ at position $g(i',x)$ in row $i'$ also ends at position $g(i',y)$. Then $G(S)$ is (semi-)repeat-free if and only if all segments of $S$ are (semi-)repeat-free.
\end{lemma}
\begin{proof}
For convenience, we say that a segment or a founder graph is \emph{valid} if it is (semi-)repeat-free, otherwise it is \emph{invalid}. Moreover, we define the following notion of a standard string occurring in $G(S) = (V,E,\ell)$.
We say that $S \in \Sigma^+$ is a \emph{standard substring of path $P = w_1 \cdots w_k$ in $G(S)$} if $P$ spells $S$ using all of its vertices, meaning
\[
    S = \ell(w_1) \big[ j..\lvert \ell(w_1) \rvert \big] \cdot \ell(w_2) \cdots \ell(w_{k-1}) \cdot \ell(w_k) \big[ 1..j' \big]
\]
with $1 \le j \le \lvert \ell(w_1) \rvert$ and $1 \le j' \le \lvert \ell(w_k) \rvert$.
We also say that the occurrence of $S$ in $P$ involves $k$ vertices of $G(S)$.

We carry out the proof of the two sides by proving their contrapositions and using the following facts:
\begin{enumerate}
    \item a segment $[x..y]$ is invalid if and only if there exist $i,i' \in \lbrace 1, \dots, m \rbrace$ such that string $\spell(\msaij{i}{x..y})$ occurs in row $\spell(\msaij{i'}{1..n})$ at some position other than $g(i',x)$, or string $\spell(\msaij{i}{x..y})$ is a proper prefix of string $\spell(\msaij{i'}{x..y})$ (for the semi-repeat-free case, ignore this last condition);
    \item founder graph $G(S)$ is invalid if and only if there exists node $v \in V$ such that $\ell(v)$ is a standard substring of some path $P = w_1 \cdots w_k$ in $G(S)$ and one of the following holds: $w_1$ is in a different block than $v$, $\ell(v)$ occurs in $\ell(P)$ at some position other than 1, or $k = 1$ and $\ell(v)$ is a proper prefix of $\ell(P) = \ell(w_1)$ (for the semi-repeat-free case, ignore this last condition).
\end{enumerate}
It is immediate to see that, by construction of $G(S)$, the additional invalidity conditions exclusive to the repeat-free case (the last conditions of facts 1.\ and 2.) are equivalent, so we concentrate on the conditions in common with the semi-repeat-free case.

($\Rightarrow$) Let $[x..y]$ be an invalid segment of $S$, with string $\spell(\msaij{i}{x..y})$ occurring in row $i'$ at some position $j$ other than $g(i',x)$, for some $i,i' \in \lbrace 1, \dots, m \rbrace$, and let $v \in V$ be the node in the block corresponding to segment $[x..y]$ such that $\ell(v) = \spell(\msaij{i}{x..y})$. If $g(i',x) < j \le g(i',y)$ then $\ell(v)$ occurs in $\ell(P)$ at position $j - g(i',x) \neq 1$, with $P$ a path starting from the same block of $v$, otherwise $j < g(i',x)$ or $j > g(i',y)$ and $\ell(v)$ occurs in some path of $G(S)$ starting from a node in a different block than that of $v$. In both cases $G(S)$ is invalid.

($\Leftarrow$) If $G(S)$ is invalid, let $\ell(v)$ be a standard substring of some path $P = w_1, \dots, w_k$ of $G(S)$ making the founder graph invalid, for some $v \in V$. Following the same arguments as in~\cite[Section 5.1]{MCENT20}, if $k \le 2$ then $\ell(P)$ is a substring of some row of the input \msa that makes $S$ invalid since by construction of $G(S)$ for every edge $(u,u') \in E$ it holds that $\ell(u)\ell(u')$ occurs in the \msa.
Otherwise $k > 2$, meaning that the occurrence of $\ell(v)$ through $P$ involves at least three vertices, and $\ell(v) = A \ell(w) B$ for some $w \in \lbrace w_2, \dots, w_{k-1} \rbrace \subseteq V$,  $A,B \in \Sigma^+$. But then $\ell(w) \in \Sigma^+$ occurs in $\ell(v)$ at some position other than 1 and so there are row indices $i,i' \in \lbrace 1, \dots, m \rbrace$ such that $\ell(w) = \spell(\msaij{i}{x..y})$ occurs in $\spell(\msaij{i'}{1..n})$ at some position other than $g(i',x)$, where $[x..y]$ is the segment of $S$ corresponding to the block of $w$ and $\spell(\msaij{i'}{1..n})$ contains $\ell(v)$, making segment $[x..y]$ invalid.
\end{proof}

\subsection{\efg construction algorithms}\label{sub:efgconstruction}
Just as in the gapless and repeat-free setting, \Cref{lem:characterization} implies that the optimal score $s(j)$ of a (semi-)repeat-free segmentation of the general \msa prefix $\msaij{1..m}{1..j}$ can be computed recursively for a variety of scoring schemes:
\begin{equation}\label{eq:score}
    s(j) =
    \bigoplus_{\substack{j' \,:\, 0 \le j' < j \;\text{s.t.} \\ \msaij{1..m}{j'+1..j} \,\text{is} \\ \text{(semi-)repeat-free}}}
    E \big(
    s(j'), j', j
    \big)
\end{equation}
where operator $\bigoplus$ and function $E$ depend on the desired scoring scheme.
Indeed:
\begin{itemize}
    \item[$i$.] for $s(j)$ to be equal to the optimal score of a segmentation maximizing the number of blocks, set $\bigoplus = \max$ and $E(s(j'),j',j) = s(j') + 1$; for a correct initialization set $s(0) = 0$ and if there is no (semi-)repeat-free segmentation set $s(j) = -\infty$;
    \item[$ii$.] for a segmentation minimizing the maximum segment length\footnote{In the gapless setting, the length of a segment and of the strings of the resulting \efg block coincide. This is not the case in the general setting, and the segment length is an upper bound to the length of the maximum node label in the resulting block.}, set $\bigoplus = \min$ and $E(s(j'),j',j) = \max ( s(j'), L([j'+1,j]) ) = \max ( s(j'), j - j' )$; set $s(0) = 0$ and if there is no (semi-)repeat-free  segmentation set $s(j) = +\infty$.
    \item[$iii$.] for a segmentation minimizing the maximum block height, set $\bigoplus = \min$ and $g(s(j'),j',j) = \max ( s(j'), W([j'+1,j]) )$; set $s(0) = 0$ and if there is no (semi-)repeat-free segmentation set $s(j) = +\infty$.
\end{itemize}

Equi et al.~\cite{Equietal22} studied the computation of semi-repeat-free segmentations optimizing for scores $i.$ and $ii$. The algorithms they developed---and that we will improve in \Cref{sect:preprocessing,sect:minmaxlength}---are based on a common preprocessing of the valid semi-repeat-free segmentation ranges, based on the following observation.
\begin{observation}[Semi-repeat-free right extensions~\cite{Equietal22}]\label{obs:rightextensions}
Given $\msaij{1..m}{1..n}$ over alphabet $\Sigma \cup \lbrace \gap \rbrace$, for any $x < y$ we say that segment $[x+1..y]$ is an extension of prefix $\msaij{1..m}{1..x}$. If extension $[x+1..y]$ is semi-repeat-free, then extension $[x+1..y']$ is semi-repeat-free for all $y < y' \le n$.
\end{observation}
\noindent Note that in the presence of gaps \Cref{obs:rightextensions} does not hold if we swap the semi-repeat-free notion with the repeat-free one, or if we swap the right extensions with the symmetrically defined left extensions.

To compute $s(j)$, \Cref{eq:score} considers all semi-repeat-free right extensions $[j'+1..j]$ ending at column $j$. Equi et al.\ discovered that the computation of values $s(j)$ can be done efficiently by considering that each semi-repeat-free right extension $[j'+1..j]$ has as prefix a minimal (semi-repeat-free) right extension $[j'+1..f(j')]$, with function $f$ defined as follows.
\begin{definition}[Minimal right extensions~\cite{Equietal22}]\label{def:rightextensions}
Given $\msaij{1..m}{1..n}$, for each $0 \le x \le n-1$ we define value $f(x)$ as the smallest integer greater than $x$ such that segment $[x+1..f(x)]$ is semi-repeat-free, or, in other words, $[x+1..f(x)]$ is the minimal (semi-repeat-free) right extension of prefix $\msaij{1..m}{1..x}$. If there is no semi-repeat-free extension, we define $f(x) = \infty$.
\end{definition}
Indeed, Equi et al.\ in \cite{Equietal22} developed an algorithm computing values $f(x)$ in time $O(mn\log m)$. Using only these values, described by a list of pairs $(x,f(x))$ sorted in increasing order by the second component, they developed two algorithms computing the score of an optimal semi-repeat-free segmentation: in time $O(n)$ for the maximum number of blocks score and in time $O(n \log\log n)$ for the maximum block length score. We will explain in detail how the latter works in \Cref{sect:minmaxlength}, as we will improve its run time to $O(n)$.

\subsection{Minimizing the maximum segment length}\label{sect:minmaxlength}
The improvement on the computation of the minimal right extensions in the case of general \msas from $O(nm \log m)$ to $O(nm)$, that we will obtain in \Cref{sect:preprocessing}, gives us the motivation to improve the $O(n \log \log n)$-time algorithm of Equi et al.~\cite[Algorithm 2]{Equietal22} for an optimal semi-repeat-free segmentation minimizing the maximum block length.
As mentioned in \Cref{sub:efgconstruction}, we can compute $s(j)$ by processing the recursive solutions corresponding to all right extensions $(x,f(x))$ with $f(x) \le j$. For the maximum block length, there are two types of recursion for an optimal solution of $\msaij{1..m}{1..j'}$ using semi-repeat-free $[x+1..j']$ as its last segment:
\begin{description}
    \item[non-leader recursion:] if $j' \le x + s(x)$ then the score of $s(j')$ is equal to $s(x)$, because the length of segment $[x+1..j']$ is less than or equal to $s(x)$; in this case, we say that $[x+1..j']$ is a \emph{non-leader segment};
    \item[leader recursion:] otherwise, if $j' > x + s(x)$, we say that $[x+1..j']$ is a \emph{leader segment}, since it gives score $j' - x$ to an optimal solution constrained to use it as its last segment.
\end{description}
\begin{figure}[htp]
\centering
\includegraphics{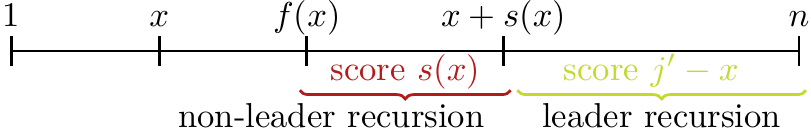}
\caption{Scheme for the score of an optimal semi-repeat-free segmentation of $\msaij{1..m}{1..j'}$ constrained to use $[x+1..j']$ as its last segment.}
\end{figure}
Note that if $x + s(x) < f(x)$ then the non-leader recursion does not occur for $(x,f(x))$. Then, it is easy to see that
\begin{equation}
    s(j) =
    \min \Bigg(
        \min_{\substack{(x,f(x)) : \\ f(x) \le j \le x + s(x)}} s(x),
        \min_{\substack{(x,f(x)) : \\ j > f(x) \;\wedge\; j > x + s(x)}} j - x
    \Bigg)
\end{equation}
so Equi et al.\ correctly solve the problem by keeping track of the two types of recursions with two one-dimensional search trees: the first keeps track of ranges $[f(x)..x+s(x)]$ with score $s(x)$, the second tracks ranges $[x+s(x)+1..n]$ where the leader recursion must be used, saving only the $-x$ part of score $j-x$.
With two semi-infinite range minimum queries, for ranges $[j+1..+\infty]$ and $[-\infty..j]$ respectively, we can compute $s(j)$ via dynamic programming and solve the problem in time $O(n \log \log n)$.

Instead, we can reach a linear time complexity using simpler data structures, thanks to the following observations:
\begin{itemize}
    \item the data structure for the leader recursion can be replaced by a single variable $S$ holding value $\min \lbrace j - x : j > f(x) \wedge j > x + s(x) \rbrace$, so that $S$ is the best score of a segmentation ending with a leader segment $[x+1..j]$;
    \item for the non-leader recursion, we can swap the structure of Equi et al.\ with an equivalent array $\mathtt{C}[1..n]$ such that $\mathtt{C}[k]$ counts the number of available solutions with score $k$ using the non-leader recursion so that a variable $K = \min \lbrace k : \mathtt{C}[k] > 0 \rbrace$ is equal to the best score of a segmentation ending with a non-leader segment $[x+1..j]$.
\end{itemize}
The final and crucial observation is that the two types of recursion are closely related: when $[x+1..j]$ goes from being a non-leader segment to a leader segment, that is, $j = x + s(x) + 1$, we decrease $\mathtt{C}[s(x)]$ by one and update $S$ with value $s(x) + 1 = j - x$ if needed. Therefore, when the best score of $\mathtt{C}[1..n]$ is removed in this way, we do not need to update $K$ to $\min \lbrace  k : \mathtt{C}[i] > 0 \rbrace$, but it is sufficient to increment $K$ by $1$ to ensure that $s(j) = \min (K,S)$, unless other updates of $\mathtt{C}$ and $S$ result in a better score.
The resulting solution is implemented in \Cref{alg:minmaxlength}.

\begin{algorithm}
\KwIn{Minimal right extensions $(x_1, f(x_1)), \dots, (x_n, f(x_n))$ sorted from smallest to largest order by the second component}
\KwOut{Score of an optimal semi-repeat-free segmentation minimizing the maximum block length}
Initialize array $\mathtt{C}[1..n]$ with values in $\lbrace 0, \dots, n \rbrace$ and set all values to $0$\;
Initialize array $\mathtt{L}[1..n]$ as empty linked-lists with values pointers to $(x, f(x))$\;
$\mathtt{minmaxlength}[0] \gets 0$\;
$y \gets 1; \; K \gets 1; \; S \gets \infty$\;
\For{$j \gets 1 \;\KwTo\; n$}{
    \While(\tcp*[f]{Process minimal right extensions}){$j = f(x_y)$}{
        \uIf(\tcp*[f]{Non-leader recursion}){$j \le x_y + \mathtt{minmaxlength}[x_y]$}{
            $\mathtt{C}[\mathtt{minmaxlength}[x_y]] \gets \mathtt{C}[\mathtt{minmaxlength}[x_y]] + 1$\;
            $K.\mathrm{Update} ( \mathtt{minmaxlength}[x_y] )$\;
            $\mathtt{L}[x_y + \mathtt{minmaxlength}[x_y] + 1].\mathrm{add}\big( (x_y, f(x_y)) \big)$\;
        }
        \Else(\tcp*[f]{Leader recursion}){
            $S.\mathrm{Update}( j - x_y )$\;
        }
        $y \gets y + 1$\;
    }
    \For(\tcp*[f]{Update non-leader rec.\ into leader rec.}){$(x,f(x)) \in \mathtt{L}[j]$}{
        $\mathtt{C}[\mathtt{minmaxlength}[x]] \gets \mathtt{C}[\mathtt{minmaxlength}[x]] - 1$\;
        $S.\mathrm{Update}(j - x)$ \tcp*{$j - x = s(x) + 1$}
    }

    \uIf{$\mathtt{C}[K] > 0$}{
        $\mathtt{minmaxlength}[j] \gets \min(K, S)$\;
    }
    \Else{
        $\mathtt{minmaxlength}[j] \gets S$\;
    }

    $S \gets S + 1$ \tcp*{Update the data structures for next iteration}
    \If{$\mathtt{C}[K] = 0$}{
        $K \gets K + 1$\;
    }
}
\Return{$\mathtt{minmaxlength}[n]$}\;
\caption{Main algorithm to find the optimal score of a semi-repeat-free segmentation minimizing the maximum block length. Operations $K.\mathrm{Update}$ and $S.\mathrm{Update}$ replace the current value of the variable with the input, if the input is smaller.}\label{alg:minmaxlength}
\end{algorithm}

\begin{theorem}
Given the minimal right extensions $(x,f(x))$ of $\msaij{1..m}{1..n}$, we can compute in time $O(n)$ the score of an optimal semi-repeat-free segmentation minimizing the maximum block length.
\end{theorem}
\begin{proof}
The correctness of \Cref{alg:minmaxlength}
follows from that of \cite[Algorithm 2]{Equietal22} and from the fact that when $\mathtt{C}[K] = 0$ we have that $\mathtt{C}[j'] = 0$ for $1 \le j' \le K$ and $S \le K + 1$. Similarly, the processing of minimal right extensions $(x,f(x))$ and the dynamic management of intervals $[f(x)..s(x) + j']$ takes time $O(n)$ in total, thus the algorithm takes linear time.
\end{proof}

\Cref{alg:minmaxlength} can be easily modified to explicitly compute an optimal segmentation. Indeed, in the main loop of \Cref{alg:mainalg}, we can keep for each $\mathtt{C}[i]$ value $\mathtt{backtrack}_\mathtt{C}[i]$ equal to the largest $x$ such that $[x+1..j]$ results in a non-leader recursive solution with score $i$.
Analogously, we can maintain value $\mathtt{backtrack}_S$ equal to the largest $x$ such that $[x+1..j]$ results in an optimal solution with score $S$ of type leader.
Then, for each $s(j)$ we can compute value $\mathtt{backtrack}[j]$ equal to the largest $x$ such that $[x+1..j]$ is used in a solution of optimal score $s(j)$.
Moreover, combined with \Cref{teo:minimalrightextensions} of \Cref{sub:minimalrightextensions}, we get the following result.
\begin{corollary}\label{corol:minmaxlength}
Given $\msaij{1..m}{1..n}$ from $\Sigma \cup \lbrace \gap \rbrace$, with $\Sigma = [1..\sigma]$ and $\sigma \le mn$, the construction of an optimal semi-repeat-free segmentation minimizing the maximum block length can be done in time $O(mn)$.
\end{corollary}

\subsection{Minimizing the maximum height in the gapless setting}\label{sub:gaplesssolution}
For gapless \msas, an $O(mn)$ solution for the construction of segmentations minimizing the maximum block height has been found by Norri et al.~\cite{NCKM19} for the case where the length of a block is limited by a given lower bound $L$, rather than with the repeat-free property.
This result holds under the assumption that $\Sigma$ is an integer alphabet of size $O(m)$.
In this section, we combine the algorithm by Norri et al.\ with the computation of values $v(j)$---that we call the \emph{minimal left extensions}---by M\"akinen et al.~\cite{MCENT20}, obtaining a linear-time solution to the construction of repeat-free founder graphs minimizing the maximum block height.

\begin{observation}[Monotonicity of left extensions~\cite{NCKM19,EGMT19}]\label{obs:left-monotonicity}
Given a gapless $\msaij{1..m}{1..n}$, for any $1 \le x \le y \le n$ we say that $[x..y]$ is a \emph{left extension} of suffix $\msaij{1..m}{y+1..n}$. Then:
\begin{itemize}
    \item if $[x..y]$ is repeat-free then $[x'\!..y]$ is repeat-free for all $x' < x$;
    \item $m \ge H([x'\!..y]) \ge H([x..y])$ for all $x' < x$.
\end{itemize}
Thus, for each $j \in [1..n]$ we define value $v(j)$ as the greatest column index smaller or equal to $j$ such that $[v(j)..j]$ is repeat-free, and we say that $v(j)$ or $[v(j)..j]$ is the \emph{minimal left extension} of $\msaij{1..m}{j+1..n}$. If there is no valid left extensions then $v(j) = -\infty$.
\end{observation}
\begin{definition}[Meaningful left extensions~\cite{NCKM19,EGMT19}]\label{def:leftextensions}
Let $\msaij{1..m}{1..n}$ be a gapless \msa. For any $j \in [1..n]$ we denote with $L_{j} = \ell_{j,1}, \dots, \ell_{j,c_j}$ the \emph{meaningful (repeat-free) left extensions} of $\msaij{1..m}{j+1..n}$, meaning the strictly decreasing sequence of all positions smaller than or equal to $j$ such that:
\begin{enumerate}
    \item $\ell_{j,c_j} < \dots < \ell_{j,2} < \ell_{j,1} = v(j)$, so that $L_j$ captures all repeat-free left extensions of $\msaij{1..m}{j+1..n}$;
    \item $H([\ell_{j,k}..j]) > H([\ell_{j,k} + 1..j])$ for $2 \le k \le c_j$, so that each $\ell_{j,k}$ marks a column where the height of the left extension increases; it follows from \Cref{obs:left-monotonicity} that $\lvert L_j \rvert = c_j \le m$.
\end{enumerate}
If $\msaij{1..m}{j+1..n}$ has no repeat-free left extension, we define $L_j = ()$ and $c_j = 0$.
Otherwise, for completeness, we define $\ell_{j,c_j+1} = -1$.
\end{definition}

Under score $iii.$ \Cref{eq:score} can be rewritten using $L_j = \ell_{j,1}, \dots, \ell_{j,c_j}$ as follows:
\begin{equation}\label{eq:left-segmentation}
    s(j) =
    \min_{k \in [1..c_j]}
    \max\bigg(
    \min_{j' \in [\ell_{j,k+1} + 1..\ell_{j,k}]} s(j'),\quad
    H \big( [\ell_{j,k}..j] \big)
    \bigg)
\end{equation}
and $s(j) = +\infty$ if $c_j = 0$,
so values $L_j$, $H([\ell_{j,k}..j])$, and $\min_{j' \in [\ell_{j,k+1}+1..\ell_{j,k}]} s(j')$ for $k \in [1..c_j]$ make it possible to compute $s(j)$ in $O(m)$ time. On one hand, given a fixed length $L$, Norri et al.~\cite{NCKM19} developed an algorithm to compute these values under the variant of \Cref{def:leftextensions} considering segments of length at least $L$---instead of repeat-free segments---in $O(mn)$ total time.
On the other hand, M{\"a}kinen et al.~\cite{MCENT20} developed a linear-time algorithm to compute values $v(j)$ of a gapless \msa.
The two solutions can be combined by finding values $v(j)$ with the latter, and by using the values as a dynamic lower bound on the minimum accepted segment length.
Since the algorithm we develop in \Cref{sub:mainalgorithm} for the general setting also solves this problem, using the symmetrically defined right extensions, we will not describe such modification in this paper.
\begin{theorem}\label{theo:gaplessheight}
Given a gapless $\msaij{1..m}{1..n}$ from an integer alphabet $\Sigma$ of size $O(m)$, an optimal repeat-free segmentation of $\msaij{1..m}{1..n}$ minimizing the maximum block height can be computed in time $O(mn)$.
\end{theorem}

\subsection{Revisiting the linear time solution for right extensions}\label{sub:mainalgorithm}
For \msas with gaps and under the semi-repeat-free notion, the monotonicity of left extensions (\Cref{obs:left-monotonicity}) fails \cite[Table 1]{Equietal22}: fixing $j \in [1..n]$, left-extensions $\msaij{1..m}{x..y}$ are not always semi-repeat-free, or \emph{valid}, from $x = v(j)$ backward, and their height could decrease when extending a valid segment.
For example, in the \msa of \Cref{fig:segmentation}, segment $[5..9]$ is semi-repeat-free but segment $[4..9]$ is not, and $H([5..9]) < H([6..9])$.
In this section, we resolve the former of the two issues, developing an algorithm exploiting right extensions and computing the optimal \msa segmentation from left to right, in the same fashion as \cite[Algorithms 1 and 2]{Equietal22} and \Cref{alg:minmaxlength}.
We will discuss the complexity of computing these right extensions in \Cref{sub:complexity}.

Recall the properties of valid right extensions and the definition of values $f(x)$ (\Cref{obs:rightextensions} and \Cref{def:rightextensions}).
\begin{definition}[Meaningful right extensions]\label{def:meaningfulrightextensions}
Given general $\msaij{1..m}{1..n}$, for any $x \in [0..n-1]$ we denote with $R_{x} = r_{x,1}, \dots r_{x,d_x}$ the \emph{meaningful (semi-repeat-free) right extensions} of $\msaij{1..m}{1..x}$, meaning the strictly increasing sequence of all positions greater than $x$ such that:
\begin{itemize}
    \item $f(x) = r_{x,1} < r_{x,2} < \dots < r_{x,d_x}$, so that $R_x$ captures all semi-repeat-free right extensions of $\msaij{1..m}{1..x}$;
    \item $H([x+1..r_{x,k}]) \neq H([x+1..r_{x,k} - 1])$ for $2 \le k \le d_x$, so that each $r_{x,k}$ marks a column where the height of the right extensions \emph{changes}.
\end{itemize}
If $\msaij{1..m}{1..x}$ has no semi-repeat-free right extension, then $R_x = ()$ and $d_x = 0$. Otherwise, for completeness, we define value $r_{x,d_x + 1} = n + 1$.
\end{definition}
Since we will treat all $R_0$, \dots, $R_{n-1}$ together, we complement each value $r_{x,k}$ with column $x$ and the height of the corresponding \msa segment, obtaining triple $(x, r_{x,k}, H([x+1..r_{x,k}]))$.

Thus, under score $iii$.\ \Cref{eq:score} can be rewritten as follows:
\begin{equation}\label{eq:rightscore}
    s(j) =
    \min_{\substack{x \in [0..j-1], \, k \in [1..d_{x}] \,: \\ r_{x,k} \le j < r_{x,k+1}}}
    \max\Big(
    s(x),\;\;
    H \big( [x+1..r_{x,k} ] \big)
    \Big).
\end{equation}
Since each $R_{x}$ defines non-overlapping ranges $[r_{x,k}..r_{x,k+1} - 1]$ over $[1..n]$, at most one range $[r_{x,k}..r_{x,k+1}-1]$ per $R_{x}$ with $x < j$ is involved in the computation of $s(j)$, and the corresponding score depends on which range contains $j$.
Also, note that \Cref{eq:rightscore} is simpler than \Cref{eq:left-segmentation}.
Finally, the algorithm computing the score of an optimal semi-repeat-free segmentation minimizing the maximum block height is described in \Cref{alg:mainalg}, and it works by processing all meaningful right extensions in $R_0, \dots, R_{n-1}$ expressed as triples $(x,r,h)$ and sorted from smallest to largest by the second component.
\begin{algorithm}
\caption{Main algorithm to find the optimal score of a semi-repeat-free segmentation minimizing the maximum block height.
Operation $\min$ returns $+\infty$ if given the empty set in input.}\label{alg:mainalg}
\KwIn{Meaningful right extensions $(x_1, r_1, h_1), \dots, (x_k, r_k, h_k)$ sorted from smallest to largest by the second component.}
\KwOut{Score of an optimal semi-repeat-free segmentation minimizing the maximum block height.}
Initialize array $\mathtt{R}[0..n-1]$ with values in $[0..m] \cup \lbrace \perp \rbrace$ and set all values to $\perp$\;
Initialize array $\mathtt{C}[1..m]$ with values in $[0..m]$ and set all values to $0$\;
$y \gets 1$\;
$\mathtt{minmaxheight}[0] \gets 0$\;
\For{$j \gets 1 \;\KwTo\; n$}{
    \While{$j = r_y$}{
        \If{$\mathtt{R}[x_y] \neq \;\perp$}{
            $\mathtt{C}[\mathtt{R}[x_y]] \gets \mathtt{C}[\mathtt{R}[x_y]] - 1$\tcc*[r]{Remove last solution of $R_{x_y}$}
        }
        $s \gets \max( \mathtt{minmaxheight}[x_y], h_y )$\;
        $\mathtt{R}[x_y] \gets s$\tcc*[r]{Save score corresponding to $(x_y,r_y,h_y)$}
        $\mathtt{C}[s] \gets \mathtt{C}[s] + 1$\tcc*[r]{Add solution corresponding to $(x_y,r_y,h_y)$}
        $y \gets y + 1$\;
    }
    $\mathtt{minmaxheight}[j] \gets \min \lbrace i \in [1..m] : \mathtt{C}[i] > 0 \rbrace$\;
}
\Return{$\mathtt{minmaxheight}[n]$}\;
\end{algorithm}

The main strategy is to keep at each iteration $j$ the best scores of the semi-repeat-free segmentations of $\msaij{1..m}{1..j}$ ending with a right extension $[1..j]$, $[2..j]$, \dots, or $[j..j]$ described by ranges in $R_0$, $R_1$, \dots, or $R_{j-1}$.
Checking each currently valid range individually would result in a quadratic-time solution, so we need to represent these ranges in some other form.
Indeed, by counting these scores with an array $\mathtt{C}[1..m]$ such that $\mathtt{C}[i]$ is equal to the number of available solutions having score $i$, score $s(j)$ can be computed by finding the smallest $i$ such that $\mathtt{C}[i]$ is greater than zero.
Array $\mathtt{C}$ needs to be updated only when $j$ reaches some $r_{x,k}$; in other words, when $j = r$ for some $(x,r,h)$, the score $\max ( s(x), h )$ of an optimal segmentation ending with $[x+1..j]$ must be added to $\mathtt{C}$, and the old score relative to the previous range of $R_{x}$ must be removed.
We can keep track of the scores in an array $\mathtt{R}[0..n-1]$ such that $\mathtt{R}[x]$ is equal to the score associated with the currently valid extension of $R_{x}$.
To compute the actual segmentation, instead of just its score, we can use two backtracking arrays $\mathtt{backtrack}[1..n]$ and $\mathtt{backtrack}_\mathtt{C}[1..m]$: $\mathtt{backtrack}_\mathtt{C}[i]$ is equal to some $x$ such that $[r_{x,k}..r_{x,k+1}-1]$ is a currently valid range resulting in score $i$ and with maximum $r_{x,k+1}$; values of $\mathtt{backtrack}_\mathtt{C}$ can be used to compute $\mathtt{backtrack}[j]$, equal to some $x$ where $[x+1..j]$ is the last segment of an optimal solution for $\msaij{1..m}{1..j}$, with maximum $r_{x,k+1}$.
Then, $\mathtt{backtrack}$ reconstructs an optimal segmentation.
\Cref{alg:mainalg} can be easily modified to update these arrays, if each meaningful right extension $(x, r_{x,k}, H([x+1..r_{x,k}]))$ is augmented with value $r_{x,k+1} - 1$.

\begin{lemma}\label{lem:mainalgorithm}
Given $\msaij{1..m}{1..n}$ and its meaningful right extensions $R_0$, $R_1$, \dots, $R_{n-1}$, we can compute the optimal semi-repeat-free segmentation minimizing the maximum block height in time $O( mn + R )$, with $R \coloneqq \sum_{x=0}^{n-1} \lvert R_x \rvert$.
\end{lemma}
\begin{proof}
The correctness follows from \Cref{eq:rightscore} and from the arguments above.
Sorting the meaningful right extensions $(x,r,h)$ by their second component can be done in time $O(n + R)$, as the meaningful right extensions take value in $[1..n]$.
Moreover, the management of arrays $\mathtt{R}$ and $\mathtt{C}$ takes constant time per meaningful right extensions, and the computation of each $s(j)$ from $\mathtt{C}$ takes $O(m)$ time, reaching the time complexity of $O(mn + R)$.
\end{proof}
For the gapless case, this is an alternative solution to that of \Cref{sub:gaplesssolution}, since $R \in O(mn)$ and the algorithms by Norri et al.\ and Equi et al.\ can be used to compute the meaningful right extensions.

\subsection{Minimizing the maximum block height in the general setting}\label{sub:complexity}
As \Cref{lem:mainalgorithm} states, we can process the meaningful right extensions of \Cref{def:meaningfulrightextensions} to compute the score of an optimal segmentation minimizing the maximum block height.
Unfortunately, in the general setting with gaps, the total number of meaningful right extensions is $O(n^2)$: as it can be seen in \Cref{fig:widthchange}, if any two row suffixes starting from the same column $x$ spell the same string but the spelling is interleaved by gaps, then the height of segment $[x..y]$ can change at any column $y$; this pattern could involve any two rows in any segment of a general $\msaij{1..m}{1..n}$, so in the worst case the quadratic upper bound is tight.

\begin{figure}[htp]
\centering
\includegraphics{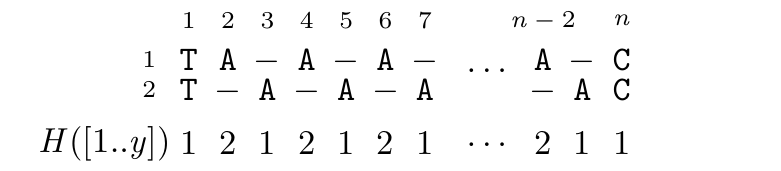}
\caption{Example of $\msaij{1..2}{1..n}$ such that $\lvert R_0 \rvert \in \Theta(n)$.
}\label{fig:widthchange}
\end{figure}
\begin{observation}\label{obs:widthdecrease}
Given $\msaij{1..m}{1..n}$ over alphabet $\Sigma \cup \lbrace\gap\rbrace$, we have that $H([x..y]) > H([x..y+1])$ only if there exist rows $i,i' \in [1..m]$ such that
$\msaij{i}{y+1} = \gap$, $\msaij{i'}{y+1} = c$, and $\spell(\msaij{i}{x..y}) = \spell(\msaij{i'}{x..y+1}) = S \cdot c$, with $c \in \Sigma$ and $S \in \Sigma^+$.
\end{observation}

The example of \Cref{fig:widthchange} and the context described by \Cref{obs:widthdecrease} seem intuitively artificial, as a high-scoring \msa would try to align the rows to avoid such a situation.
Nonetheless, without further assumptions about gaps in the \msa portions reading the same strings, we are left to compute all meaningful right extensions.
Indeed, let $K_{x,y}$ be the keyword tree of the set of strings $S_{x..y} \coloneqq \lbrace \spell(\msaij{i}{x..y}) : 1 \le i \le m \rbrace$; since $\lvert S_{x..y} \rvert = H([x..y])$, the height of $[x..y]$ is equal to the number of distinct nodes of $K_{x..y}$ corresponding to the strings in $S_{x..y}$.
We can obtain a parameterized solution by noting that if $K_{x..y}$ has $m$ leaves then no two strings in $S_{x..y}$ are one prefix of the other and $H([x..y']) = m$ for all $y' > y$.
\begin{lemma}\label{lem:rightextensions}
Given general $\msaij{1..m}{1..n}$ over integer alphabet $\Sigma \cup \lbrace \gap \rbrace$ of size $\sigma \in O(mn)$, we denote with $\alpha$ the maximum length $y - x + 1$ of any segment $[x..y]$ such that $\spell(\msaij{i}{x..y})$ is a prefix of $\spell(\msaij{i'}{x..y})$ for some $i,i' \in [1..m]$.
Then, we can compute all meaningful right extensions in time $O(m n \alpha \log \sigma )$.
\end{lemma}
\begin{proof}
For each $x \in [0..n-1]$, we can find $R_x$ by incrementally computing trees
$K_{x+1..x+1}$, $K_{x+1..x+2}$, \dots, $K_{x+1..n}$
using a dynamic keyword tree $\mathcal{T}$ supporting the traversal from the root to the leaves and the insertion of a $c$-child to an arbitrary node $v$, with $c \in \Sigma$.
A possible implementation of the procedure is described in \Cref{alg:realrightextensions}.
During the computation, array $\mathtt{V}[1..m]$ keeps track of the nodes corresponding to strings $\spell(\msaij{i}{x+1..r})$, each variable $v.\mathrm{count}$ counts the number of rows reading the corresponding string, and a variable $h$ counts the number of distinct nodes $v$ such that $v.\mathrm{count}$ is greater than zero: if $h$ changes then the corresponding meaningful right extension of $[x+1..r]$ is $(x,r,h)$.
For each $R_x$, the algorithm stops if the number of leaves of $\mathcal{T}$ is $m$, that is it incrementally computes at most $\alpha$ keyword trees; no meaningful right extension is missed, thanks to \Cref{obs:widthdecrease} and the above arguments.
We can compute $R_0$, \dots, $R_{n-1}$ in time $O(m n \alpha \log \sigma)$, since the traversal and insertion operations of $\mathcal{T}$ can be implemented in time $O(\log \sigma)$\footnote{Since only insertion is needed and alphabet $\Sigma$ is fixed, the addition of a child to each node $v$ can be implemented with a dynamic binary tree with height at most $\lceil \log_2 \sigma \rceil$ leaves, growing downwards as the number of children grows.} and the other operations can be supported in constant time.
\end{proof}

\begin{algorithm}
\caption{Algorithm computing all meaningful right extensions $R_0$, \dots, $R_{n-1}$. To efficiently compute keyword trees $K_{x+1..r}$ for $y \in [x+1..n]$, we need a dynamic tree data structure $\mathcal{T}$ supporting navigation and insertions in time $O(\log \sigma)$.}\label{alg:realrightextensions}
\SetKw{KwOutput}{output}
\KwIn{$\msaij{1..m}{1..n}$ from an integer alphabet $\Sigma \cup \lbrace \gap \rbrace$ of size $\sigma \in O(mn)$, minimal right extensions $f(x)$ for $x \in [0..n-1]$.}
\KwOut{Meaningful prefix-aware right extensions $R_0$, \dots, $R_{n-1}$ represented as triples $\big(x,r_{x,k},H([x+1..r_{x,k}])\big)$.}
\For{$x \gets 0 \;\KwTo\; n-1$}{
    Initialize empty keyword tree $\mathcal{T}$, containing only node root\;
    Initialize array $\mathtt{V}[1..m]$ with values pointers to nodes of $\mathcal{T}$ and set all values to node $\rroot$\;
    $h \gets 0$\;
    $r \gets x$\;
    \While{$\mathcal{T}.\mathrm{leaves} < m \;\wedge\; r \le m$}{
        $h' \gets h$\;
        \For{$i \gets 1 \;\KwTo\; m$} {
            \If{$\msaij{i}{r} \neq \gap$}{
                $\mathtt{V}[i].\mathrm{count} \gets \mathtt{V}[i].\mathrm{count} - 1$\;
                \If{$\mathtt{V}[i].\mathrm{count} = 0$}{
                    $h \gets h - 1$\;
                }
                \uIf{$\mathtt{V}[i]$ has an $(\msaij{i}{r})$-child $v$}{
                    $\mathtt{V}[i] \gets v$\;
                }
                \Else{
                    Add new $(\msaij{i}{r})$-child $v$ to $\mathtt{V}[i]$\;
                    $\mathtt{V}[i] \gets v$\;
                }
                \If{$\mathtt{V}[i].\mathrm{count} = 0$}{
                    $h \gets h + 1$\;
                }
                $\mathtt{V}[i].\mathrm{count} \gets \mathtt{V}[i].\mathrm{count} + 1$\;
            }
        }
        \uIf{$r = f(x)$}{
            $\KwOutput\; (x,r,h)$\;
        }
        \ElseIf{$r \ge f(x) \;\wedge\; h \neq h'$}{
            $\KwOutput\; (x,r,h)$\;
        }
        $r \gets r + 1$\;
    }
}
\end{algorithm}

Since the total number of meaningful right extensions is $O(mn \alpha)$, \Cref{lem:rightextensions,lem:mainalgorithm} give the following solution to our segmentation problem.
\begin{theorem}\label{theo:generalheight}
Given general $\msaij{1..m}{1..n}$ over an integer alphabet $\Sigma$ of size $O(mn)$, we can compute the score of an optimal segmentation minimizing the maximum block height in time $O(m n \alpha \log \sigma)$, where $\alpha$ is the length of the longest \msa segment where any two rows spell strings $S,S'$ such that $S$ is a prefix of $S'$.
\end{theorem}

In the worst case, the number of meaningful right extensions is $O(mn^2)$, $\alpha \in \Theta(n)$, and the time complexity of \Cref{lem:rightextensions} is $\Theta(mn^2 \log \sigma)$.
Thus, we introduce a different generalization of block height from the gapless setting to the general one.
\begin{definition}[Prefix-aware height]\label{def:prefixawareheight}
Given $\msaij{1..m}{1..n}$, we define the \emph{prefix-aware height} of a segment $[x..y]$, denoted as $\pah(\msaij{1..m}{x..y})$ or just $\pah([x..y])$, as the number of distinct strings $S$ in $\lbrace \spell(\msaij{i}{x..y}) : 1 \le i \le m \rbrace$ such that $S$ is not a prefix of some other string of the set.
\end{definition}
Since $\pah([x..y])$ is equal to $H([x..y])$ minus the number of strings spelled in $[x..y]$ that are proper prefixes of other strings of the segment, this refined height is always smaller or equal to the original height: the relative optimal segmentation provides a lower bound for the maximum height in the original setting.
Moreover, the necessary condition for the decrease in height stated in \Cref{obs:widthdecrease} is no longer valid, and it is easy to see that the monotonicity of \emph{prefix-aware} right extensions holds (see \Cref{obs:left-monotonicity}).
Indeed, if we define the \emph{meaningful prefix-aware right extensions} $\overline{R}_0$, \dots, $\overline{R}_{n-1}$ as in \Cref{def:meaningfulrightextensions}, it is easy to see that $\lvert \overline{R}_x \rvert \le m + 1$ for all $x \in [0..n-1]$, so the number of these extensions is $O(mn)$ in total.

\begin{definition}[Meaningful prefix-aware right extensions]\label{def:mpfarightextensions}
Given $\msaij{1..m}{1..n}$ over alphabet $\Sigma \cup \lbrace \gap \rbrace$, for any $x \in [0..n-1]$ we denote with $\overline{R}_{x} = r_{x,1}, \dots r_{x,d_x}$ the \emph{meaningful prefix-aware right extensions} of $\msaij{1..m}{1..x}$ as in \Cref{def:meaningfulrightextensions}, substituting segment height $H$ with prefix-aware height $\overline{H}$.
\end{definition}

Finally, given $\overline{R}_0$, \dots, $\overline{R}_{n-1}$ as input, \Cref{alg:mainalg} correctly computes the score of an optimal segmentation under our refined height, since \Cref{eq:rightscore} still holds.
In \Cref{sect:preprocessing}, we will provide an algorithm based on the generalized suffix tree of the gaps-removed \msa rows computing the prefix-aware extensions in time linear in the \msa size, obtaining the following result.
\begin{theorem}\label{theo:pah}
Given $\msaij{1..m}{1..n}$ over integer alphabet $\Sigma \cup \lbrace \gap \rbrace$ of size $\sigma \le mn$, computing a semi-repeat-free segmentation minimizing the maximum prefix-aware block height takes $O(mn)$ time.
\end{theorem}

\section{Preprocessing the \msa for its segmentation}\label{sect:preprocessing}
In this section, we study the computation of the minimal right extensions $f(x)$ (\Cref{def:rightextensions}) and the meaningful prefix-aware right extensions $\overline{R}_x$ (\Cref{def:mpfarightextensions}), for $0 \le x \le n - 1$.
Our goal is to compute them in time $O(n)$ and $O(mn)$, respectively.
For the former, Equi et al.\ in \cite{Equietal22} proposed an $O(nm \log m)$-time solution using the following data structure, built from the gaps-removed \msa rows.

\begin{definition}
Given $\msaij{1..m}{1..n}$ from alphabet $\Sigma \cup \lbrace \gap \rbrace$, we define \gst as the generalized suffix tree of the set of strings
$\lbrace \spell( \msaij{i}{1..n} ) \cdot \$_i : 1 \le i \le m \rbrace$, with $\$_1, \dots, \$_m$ $m$ new distinct terminator symbols not in $\Sigma$.\footnote{We added the $m$ new distinct terminators for simplicity, whereas Equi et al.\ used the suffix tree of the concatenation of all gaps-removed rows with a single new symbol $\$$ between each. The suffix tree of this string, if a second unique terminator $\#$ is concatenated to this string, is equivalent to \gst for our purposes.}
\end{definition}
An example of \gst is given in \Cref{fig:gstmsa}. From the suffix tree properties, it follows that for any gaps-removed row $S_i \coloneqq \spell(\msaij{i}{1..n}) \$_i$, with $1 \le i \le m$:
each suffix $S_i[x..\lvert S_i \rvert]$ corresponds to a unique leaf $\ell_{i,x}$ of \gst and vice versa, with $1 \le x \le \lvert S_i \rvert$;
each substring $S_i[x..y]$ corresponds to an explicit or implicit node of \gst in the root-to-$\ell_{i,x}$ path; and each explicit or implicit node corresponds to one or more of such substrings, uniquely identifiable thanks to the leaves covered by the node.
Also, note that \gst does not contain any information about the gap symbols of the \msa, as this information will be added back into the structure thanks to the set of leaves and nodes considered.

\begin{figure}[htp]
\centering
\includegraphics{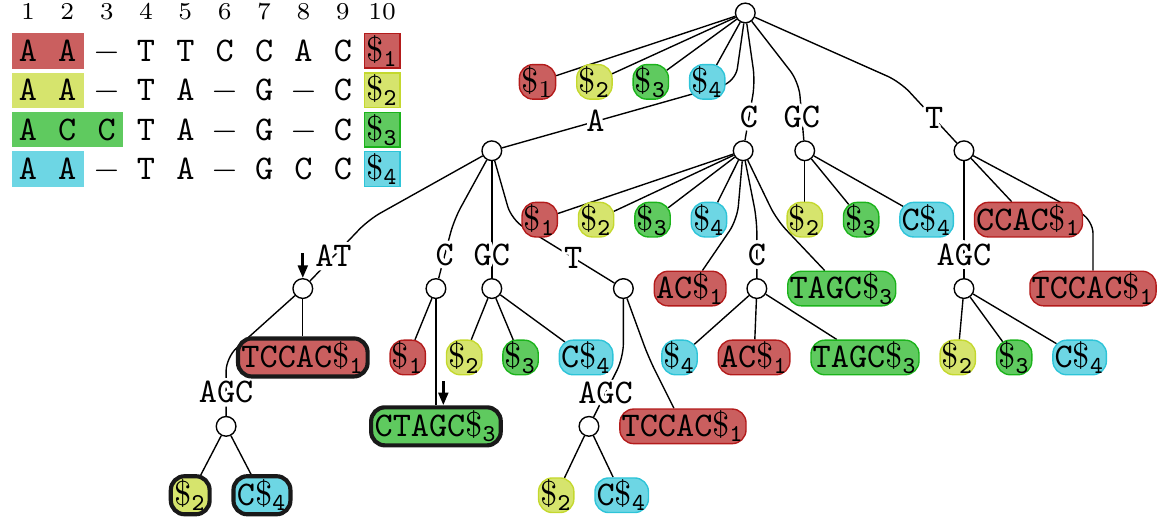}
\caption{Example of an $\msaij{1..4}{1..10}$ and its \gst, where the label to each leaf has been moved inside the leaf itself.
We have also highlighted the leaves corresponding to suffixes $\spell(\msaij{i}{1..n})$ (black outline) and its exclusive ancestors (arrows).
}\label{fig:gstmsa}
\end{figure}

In \Cref{sub:gst} we perform an analysis of \gst similar to that of Equi et al., showing that semi-repeat-free segments of the \msa correspond to a specific set of nodes of \gst covering exactly $m$ leaves. Then, in \Cref{sub:exclusiveancestorset}, we show that the novel resulting problem on the tree structure of \gst, that we call the \emph{exclusive ancestor set problem}, can be solved efficiently, resulting in an algorithm computing the minimal right extensions in linear time, a solution that we describe in \Cref{sub:minimalrightextensions}.
Moreover, in \Cref{sec:paextensions}, we extend these techniques to show that the forests inside \gst identified by the exclusive ancestors can describe the meaningful prefix-aware right extensions: by computing for each node of these forests the position indicating where the first \msa occurrence of the related string ends, and by sorting these positions, we can compute the meaningful right extension in $O(m^2n)$ global time.
Finally, in \Cref{sub:pos} we describe how these positions can be computed efficiently, thanks to the \emph{generalized prefix tree} of the gaps-removed rows and a map between the suffix tree and prefix tree nodes.
This map is computable in linear time, as an application of affix trees or affix arrays \cite{DBLP:journals/algorithmica/Maass03,DBLP:journals/tcs/Strothmann07}, or with the data structure for weighted ancestor queries of Belazzougui et al.\ \cite{DBLP:conf/cpm/BelazzouguiKPR21}, making it possible to navigate from the suffix tree to the prefix tree in constant time, reaching global $O(mn)$ time.

\subsection{Semi-repeat-free segments in the generalized suffix tree}\label{sub:gst}
The following has been stated and exploited in~\cite{Equietal22}.
\begin{definition}[Semi-repeat-free substrings]\label{def:semirepeatfreesubstrings}
Recall the definition of a semi-repeat-free segment (\Cref{lem:characterization}). Given substring $\msaij{i}{x..y}$ of $\msaij{1..m}{1..n}$ such that $\spell(\msaij{i}{x..y}) \in \Sigma^+$, we say that $\msaij{i}{x..y}$ is a \emph{semi-repeat-free substring} if for all $1 \le i' \le m$ string $\spell(\msaij{i}{x..y})$ occurs in gaps-removed row $i'$ only at position $g(i',x)$ (or it does not occur at all).
\end{definition}
\begin{observation}\label{obs:rowsegmentf}
Segment $[x..y]$ is semi-repeat-free if and only if all substrings $\msaij{i}{x..y}$ are semi-repeat-free, for $1 \le i \le m$.
If $\msaij{i}{x..y}$ is semi-repeat-free, then $\msaij{i}{x..y'}$ is semi-repeat-free for all $y < y' \le n$.
Recall the definition of minimal right extension $f(x)$, for $0 \le x \le n-1$ (\Cref{def:rightextensions}).
Let $f^i(x)$ be the smallest integer greater than $x$ such that substring $\msaij{i}{x+1..f^i(x)}$ is semi-repeat-free: it is easy to see that $f(x) = \max_{i=1}^{m} f^i(x)$.
\end{observation}

This translates into a specific set of implicit or explicit nodes of \gst. 
The fact that we added a unique terminator symbol to each row is equivalent to the addition of an \msa column spelling $\$_1 \cdots \$_m$ at position $n+1$, which means that $[x+1..n+1]$ is always semi-repeat-free and the minimal right extensions such that $f(x) = \infty$ become $f(x) = n + 1$.
\begin{lemma}\label{lem:gst}
Given $m$ row substrings $\msaij{i}{x..y_i}$ of $\msaij{1..m}{1..n}$ such that $\spell(\msaij{i}{x..y_i}) \in \Sigma^+$ for $1 \le i \le m$,
let $W = \lbrace w_1, \dots, w_k \rbrace$ be the set of implicit or explicit nodes of \gst corresponding to strings
$\lbrace \spell( \msaij{i}{x..y_i} ) : 1 \le i \le m \rbrace$.
Then
    $\msaij{i}{x..y_i}$ is semi-repeat-free for all $1 \le i \le m$
    if and only if
    $W$ covers exactly $m$ leaves in \gst.
\end{lemma}
\begin{proof}
By construction of \gst, $W$ covers the $m$ leaves $\ell_{1,z_1}, \dots, \ell_{m,z_m}$, with $z_i = g(i,x)$, so we only need to prove that if some $\msaij{i}{x..y_i}$ is not semi-repeat-free, or \emph{invalid}, then $W$ covers more than $m$ leaves, and vice versa.

($\Leftarrow$) Let $\msaij{i}{x..y_i}$ be invalid, i.e.\ $\spell(\msaij{i}{x..y_i})$ occurs in $S_{i'}$ at some position $\hat{z}$ other than $z_{i'}$, for some row $1 \le i' \le m$. Then the node of \gst corresponding to string $\spell(\msaij{i}{x..y_i})$ covers leaf $\ell_{i'\!,\hat{z}} \neq \ell_{i'\!,z_{i'}}$, thus $W$ covers more than $m$ leaves.

($\Rightarrow$) Let $\ell_{i',\hat{z}}$ be a leaf of \gst other than leaves $\ell_{1,z_1}, \dots, \ell_{m,z_m}$ covered by some node $w \in W$. By construction, $w$ corresponds to $\spell( \msaij{i}{x..y_i} )$ for some $1 \le i \le m$, so we have that $\spell( \msaij{i}{x..y_i} )$ occurs in $S_{i'}$ at some position other than $g(i',x)$, since $\ell_{i'\!,\hat{z}} \neq \ell_{i'\!,z_{i'}}$. Thus, $\msaij{i'}{x..y_i}$ is invalid.
\end{proof}
Note that the correctness of \Cref{lem:gst} does not hold if we swap the semi-repeat-free notion with the repeat-free one.

\Cref{lem:gst}, combined with \Cref{obs:rowsegmentf}, implies that the problem of computing values $f^i(x)$ for all $i \in [1..m]$ can be solved by analyzing the tree structure of \gst against the \msa suffixes.
Indeed, let $\mathcal{L}_x \coloneqq \lbrace \ell_{i,z_i} : 1 \le i \le m, z_i = g(i,x+1) \rbrace$ be the leaves of \gst corresponding to the suffixes $\spell(\msaij{i}{x+1..n})$.
For each row $1 \le i \le m$, the first semi-repeat-free prefix of $\spell(\msaij{i}{x+1..n})$ corresponds to the first implicit or explicit node $v$ of \gst in the root-to-$\ell_{i,z_i}$ path such that $v$ covers only leaves in $\mathcal{L}_x$.
The fact that \gst is a compacted trie is not an issue: the parent of $v$ in the suffix trie is branching, since it covers more leaves than $v$, so the first explicit node of \gst in the root-to-$\ell_{i,z_i}$ path covering only leaves in $\mathcal{L}_x$ is the first explicit descendant $w$ of $v$, thus we can identify $v$ by finding $w$.
Finally, $f^i(x)$ is computed by retrieving the smallest column index $y$ such that $\spell(\msaij{i}{x+1..y}) = \sstring(\parent(w)) \cdot \cchar(w)$, where $\sstring(u)$ is the concatenation of edge labels of the root-to-$u$ path, and $\cchar(u)$ is the first symbol of the edge label from $\parent(u)$ to $u$.
In other words, $y$ corresponds to the $k$-th non-gap symbol of \msa row $i$, with $k = \mathrm{rank}(\msaij{i}{1..n}, x) + \stringdepth(\parent(w)) + 1$, where $\mathrm{rank}(\msaij{i}{1..n}, x)$ is the number of non-gap symbols in $\msaij{i}{1..x}$ and $\stringdepth(u) \coloneqq \lvert \sstring(u) \rvert$.
For example, in \Cref{fig:gstmsa} the leaves of $\mathcal{L}_0$ have been marked and so have the shallowest ancestors covering only leaves in $\mathcal{L}_0$.

\subsection{Exclusive ancestor set}\label{sub:exclusiveancestorset}
The results of the previous section show that we can compute the minimal right extensions by solving multiple instances of the following problem on the tree structure of \gst.
\begin{problem}[Exclusive ancestor set]\label{prob:exclusiveancestorset}
Let $T = (V,E,\mathrm{root})$ be a rooted ordered tree, with $L^T \subseteq V$ the set of its leaves. Given $T$ and a subset of leaves $L \subseteq L^T$, find the minimal set $W$ of exclusive ancestors of $L$ in $T$, i.e.\ the minimal set $W \subseteq V$ such that $W$ covers all leaves in $L$ and only leaves in $L$.
Can $T$ be preprocessed to support the efficient solving of multiple instances of the problem?
\end{problem}

As is the case for \gst, we can assume that each internal node of $T$ has at least two children, otherwise, a linear-time processing of $T$ can be employed to compact its unary paths. Indeed, after a linear-time preprocessing of $T$, any instance of \Cref{prob:exclusiveancestorset} defined by $L$ and $T$ can be solved in time $O(\lvert L \rvert)$ by a careful traversal of the tree with the following procedure, that we describe informally:
\begin{enumerate}
    \item partition $L$ in $k$ maximal sets $L_1$, \dots, $L_k$ of leaves contiguous in the ordered traversal of $T$, to be processed independently (if two leaves belong to different contiguous sets, any common ancestor cannot be part of the solution);
    \item for each $L_i$, with $1 \le i \le k $, start from the leftmost leaf $\ell_i$ and ascend in the tree until the closest ancestor of $\ell_i$ that covers some leaf not in $L_i$;
    \item upon failure in step 2., add the last safe ancestor to the solution $W$ and if there are still uncovered leaves in $L_i$ repeat steps 2.\ and 3.\ starting from the leftmost uncovered leaf.
\end{enumerate}
An example of the procedure is shown in \Cref{fig:exclusiveancestorset}. The failure condition of step 2.\ can be evaluated by checking if both the leftmost leaf and the rightmost leaf in the subtree of the candidate replacement are still in set $L_i$, and step 2.\ always terminates if we assume that $L$ is a nontrivial instance: if $L \subset L^T$, then the root of $T$ is not the solution to the problem.
\begin{figure}[htp]
\centering
\includegraphics{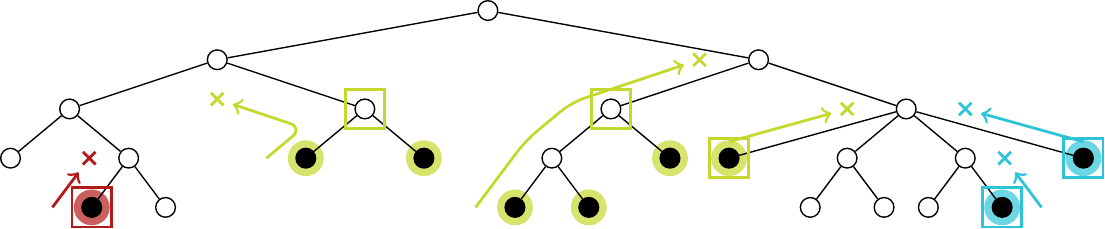}
\caption{Example of an instance of exclusive ancestor set, where the set of leaves $L$ corresponds to the black leaves: the algorithm partitions $L$ into sets of contiguous leaves (shown as red, lime, and blue leaves), and for each set, it finds the exclusive ancestors (marked with rectangles). Each arrow shows the ascent of step 2.\ up the tree until the node corresponding to the failure condition, which is marked with a cross.}\label{fig:exclusiveancestorset}
\end{figure}

Assuming the leaves of $T$ are sorted, step 1.\ can be implemented efficiently: we can partition $L$ into sets of contiguous leaves by coloring leaves in $L$ and finding all the leaves with the preceding leaf not in $L$.
We can easily preprocess $T$ to support the required operations in constant time, leading to a time complexity of $O(\lvert L \rvert)$, since any forest built on top of leaves $L$ has $O(\lvert L \rvert)$ nodes.
\begin{lemma}\label{lem:exclusiveancestorset}
The exclusive ancestor set problem on a rooted ordered tree $T = (V,E,\mathrm{root})$ and a subset $L$ of its leaves can be solved in time O$(\lvert L \rvert)$, after a $O(\lvert V \rvert)$-time preprocessing to support operations $v.\leftmostleaf$, $v.\rightmostleaf$ on any node $v \in V$ and operations $\ell.\prevleaf$, $\ell.\nextleaf$, and the binary coloring of any leaf $\ell \in L^T$ in constant time.
\end{lemma}

\subsection{Computing the minimal right extensions}\label{sub:minimalrightextensions}
Returning to the problem of computing values $f(x)$, the representation of \gst needs to support the operations on its tree structure described by \Cref{lem:exclusiveancestorset} plus operations $v.\stringdepth$, returning the length of the string corresponding to the root-to-$v$ path in \gst of an explicit node $v$, and $\ell.\suffixlink$, implementing the suffix links of the leaves.
The final algorithm, described in \Cref{alg:minimalrightextensions}, 
computes leaf sets $\mathcal{L}_0$, $\mathcal{L}_1$, \dots, $\mathcal{L}_{n-1}$ corresponding to the \msa suffixes starting at column $1, 2, \dots, n$, respectively, and for each $\mathcal{L}_x$ with $0 \le x < n$:
\begin{enumerate}
    \item it marks the leaves in $\mathcal{L}_x$ and partitions them in sets of contiguous leaves, by finding all their left boundaries $\ell$ such that $\ell.\prevleaf$ is not marked;
    \item it solves the exclusive ancestor set problem on each set of contiguous leaves and whenever it finds an exclusive ancestor, covering leaves $\ell_{i_1}, \dots, \ell_{i_k}$, it computes values $f^i(x)$ for $i \in \lbrace i_1, \dots, i_k \rbrace$ (see the conclusion of \Cref{sub:gst});
    \item after processing all leaves, it finally computes $f(x) = \max_{i=1}^{m} f^i(x)$ and transforms $\mathcal{L}_x$ into $\mathcal{L}_{x+1}$ by taking the suffix links\footnote{As noted by an anonymous reviewer for the conference version of the paper corresponding to this section~\cite{RM22b}, the support for suffix links is not strictly necessary, since we are exploring leaves only. Indeed, a traversal of the tree can easily fill an $m \times n$ table containing $\mathcal{L}_0$, \dots, $\mathcal{L}_{n-1}$, that we then have to store.} of only leaves $\ell_{i}$ such that $\msaij{i}{x+1} \neq \gap$.
\end{enumerate}

\begin{algorithm}
\SetKw{KwOutput}{output}
\SetKwRepeat{Do}{do}{while}
\KwIn{$\msaij{1..m}{1..n}$ from alphabet $\Sigma \cup \lbrace \mathtt{-} \rbrace$}
\KwOut{Pairs $(x,f(x))$ for $x = 0, \dots, n-1$}
Preprocess the \msa rows to support select and rank queries\;
Build \gst, the generalized suffix tree of $\lbrace \spell(\msaij{i}{1..n}) \$_i : 1 \le i \le m \rbrace$\;
Initialize $\mathtt{L}[i]$ as the leaf of \gst corresponding to $\spell(\msaij{i}{1..n})$\;
\For{$x \gets 0 \;\KwTo\; n-1$}{
    \For(\tcp*[f]{Mark each leaf of $\mathcal{L}_x$}){$i \gets 1 \;\KwTo\; m$}{
        $\mathtt{L}[i].\mathrm{marked} \gets \atrue$\;
    }
    \tcp{Process each set of contiguous leaves}
    \For{$\mathrm{i} \in \lbrace 1, \dots, m \rbrace : \mathtt{L}[i].\mathrm{prevleaf}.\mathrm{marked} = \afalse$}{
        $\mathrm{lb} \gets \mathtt{L}[i]$\tcc*[r]{Find left and right boundaries}
        $\mathrm{rb} \gets \mathrm{lb}$\;
        \While{$\mathrm{rb}.\mathrm{nextleaf}.\mathrm{marked} = \atrue$}{
            $\mathrm{rb} \gets \mathrm{rb}.\mathrm{nextleaf}$\;
        }
        $w \gets \mathrm{lb}; \; \mathrm{lleaf} \gets \mathrm{lb}; \;  \mathrm{rleaf} \gets \mathrm{lb}$\tcc*[r]{Find the exclusive ancestors}
        \While{$\mathrm{rleaf} \le \mathrm{rb}$}{
            $w' \gets w.\mathrm{parent}; \; \mathrm{lleaf}' \gets w'\!.\mathrm{leftmostleaf}; \; \mathrm{rleaf}' \gets w'\!.\mathrm{rightmostleaf}$\;
            \uIf(\tcp*[f]{$w'$ is a correct replacement}){$\mathrm{lb} \le \mathrm{lleaf}' \,\wedge\, \mathrm{rleaf}' \le \mathrm{rb}$}{
                $w \gets w'; \; \mathrm{lleaf} \gets \mathrm{lleaf}'; \; \mathrm{rleaf} \gets \mathrm{rleaf}'$\;
            }
            \Else(\tcp*[f]{$w'$ fails so $w$ is an exclusive ancestor}){
                $g \gets (w.\mathrm{parent}).\mathrm{stringdepth} + 1$\;
                $\ell \gets \mathrm{lleaf}$\;
                \While{$\ell \le \mathrm{rleaf}$}{
                    $i' \gets \ell.\mathrm{getrow}$\;
                    $\mathtt{f}[i'] \gets \msaij{i'}{1..n}.\mathrm{select}( \msaij{i'}{1..n}.\mathrm{rank}(x) + g )$\;
                    $\ell \gets \ell.\mathrm{nextleaf}$\;
                }
                $w \gets \ell; \; \mathrm{lleaf} \gets \ell; \;  \mathrm{rleaf} \gets \ell$\;
            }
        }
    }
    \KwOutput{$(x,\max_{i=1}^{m}\mathtt{f}[i])$}\;
    \For(\tcp*[f]{Cleanup and compute $\mathcal{L}_{x+1}$}){$i \gets 1 \;\KwTo\; m$}{
        $\mathtt{L}[i].\mathrm{marked} \gets \afalse$\;
        \If{$\msaij{i}{x+1} \neq \mathtt{-}$}{
            $\mathtt{L}[i] \gets \mathtt{L}[i].\suffixlink$\;
        }
    }
}
\caption{Algorithm computing the minimal right extensions $f(x)$, for $0 \le x \le n - 1$. 
For simplicity, we use the total order $\le$ to check the failure condition of the exclusive ancestor set problem, but it is easy to see that $\le$ can be replaced by an additional binary coloring of the leaves, to check if a leaf is contained in a set of contiguous leaves of interest.}\label{alg:minimalrightextensions}
\end{algorithm}

\begin{theorem}\label{teo:minimalrightextensions}
Given $\msaij{1..m}{1..n}$, we can compute the minimal right extensions $f(x)$ for $0 \le x < n$ (\Cref{def:rightextensions}) in time $O(mn)$.
\end{theorem}
\begin{proof}
The correctness is given by \Cref{obs:rowsegmentf} and \Cref{lem:gst,lem:exclusiveancestorset}. The construction of \gst is equivalent to building the suffix tree of a string of length smaller than or equal to $(m+1) \cdot n$: a suffix tree supporting the required operations in constant time can be constructed in $O(mn)$ time since we assume $\lvert \Sigma \rvert \le mn$.
Also, we can preprocess the \msa rows to answer in constant time rank and select queries on the position of gap and non-gap symbols. Thus, the computation of each $f(x)$ takes time $O(\lvert \mathcal{L}_x \rvert + m) = O(m)$, so $O(mn)$ time in total.
\end{proof}

\begin{corollary}\label{corol:maxblocks}
Given $\msaij{1..m}{1..n}$ from $\Sigma \cup \lbrace \gap \rbrace$, with $\Sigma = [1..\sigma]$ and $\sigma \le mn$, the construction of an optimal semi-repeat-free segmentation minimizing the maximum number of blocks can be done in time $O(mn)$.
\end{corollary}
\begin{proof}
Algorithm \cite[Algorithm 1]{Equietal22} by Equi et al.\ solves the problem in $O(n)$ time, assuming it is given the minimal right extensions $(x,f(x))$ sorted in increasing order by the second component, which we can now compute and sort in time $O(mn)$ thanks to \Cref{teo:minimalrightextensions}.
\end{proof}

\subsection{Computing the prefix-aware right extensions}\label{sec:paextensions}
Given \gst and its leaves $\mathcal{L}_x$ corresponding to the suffixes starting at column $x+1$, the forest with $m$ leaves identified by the exclusive ancestors $W_x$ of $\mathcal{L}_x$ can also be used to study the meaningful prefix-aware right extensions $\overline{R}_x$ (\Cref{def:mpfarightextensions}).

\begin{definition}[First ending position]\label{def:pos}
Given \gst, let $F_x$ be the set of all explicit nodes of \gst belonging to the subtree rooted at some exclusive ancestor $w \in W_x$.
Then, for each $v \in F_x$ we define value $\pos(v)$ as the \emph{first ending position} of string $\mathrm{string}(\parent(v)) \cdot \mathrm{char}(v)$ in the \msa, where $\mathrm{char}(v)$ is the first character of the label from $v$'s parent to $v$.
In other words, if $S = \sstring(\parent(v)) \in \Sigma^+$ and $c = \cchar(v) \in \Sigma$, then $\pos(v)$ is the minimum column index $y \in [1..n]$ such that $Sc = \spell(\msaij{i}{x+1..y})$ for some $1 \le i \le m$.
\end{definition}
An example of sets $W_x$, $F_x$, and of values $\pos(v)$ is shown in \Cref{fig:gstmsaheight,fig:pahexample}.
After plotting these values in a horizontal line, it is easy to notice that all increases in $\pah$ correspond to some $\pos(v)$, but not the other way around: the $\pos$ values that do not affect $\pah$, because they do not correspond to two or more rows reading strings that are not one prefix of the other, are the first-born children---with respect to the value of $\pos$---of branching nodes.

\begin{figure}[htp]
\centering
\includegraphics{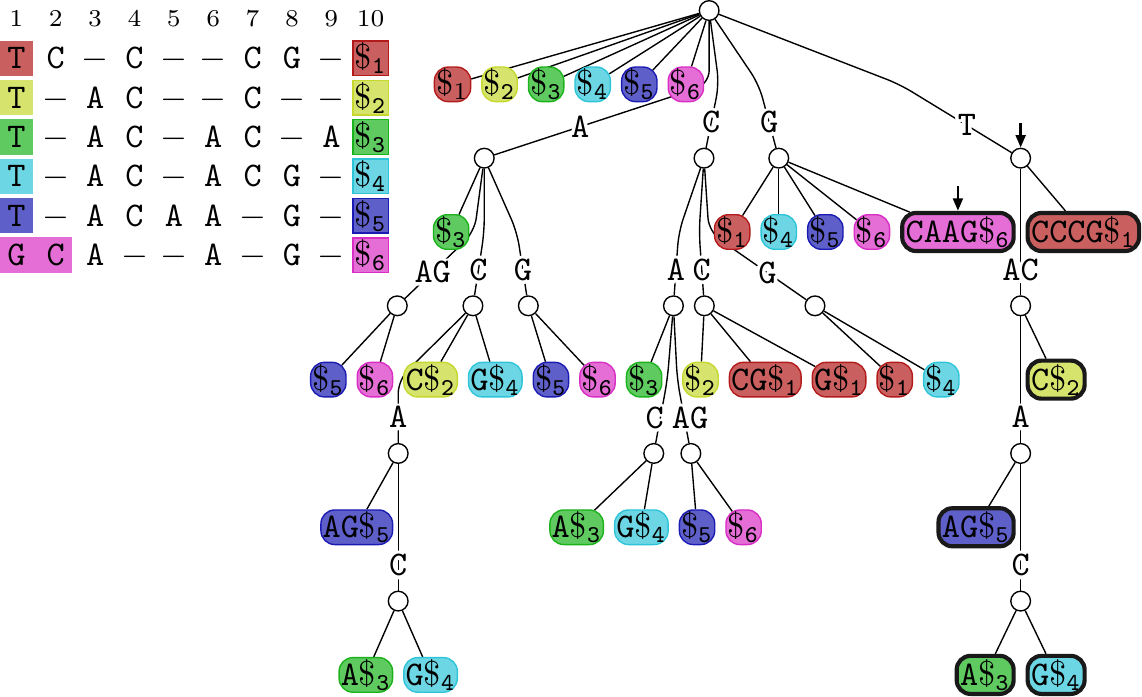}
\caption{A different example of an $\msaij{1..6}{1..9}$ and its \gst.
We have also highlighted, with a black outline, the leaves $\mathcal{L}_0$ corresponding to suffixes $\spell(\msaij{i}{1..n})$; their exclusive ancestors $W_0$, the nodes corresponding to $\mathtt{G \cdot CAAG\$_6}$ and $\mathtt{T}$, are marked with arrows.
}\label{fig:gstmsaheight}
\end{figure}
\begin{figure}[htp]
\centering
\includegraphics{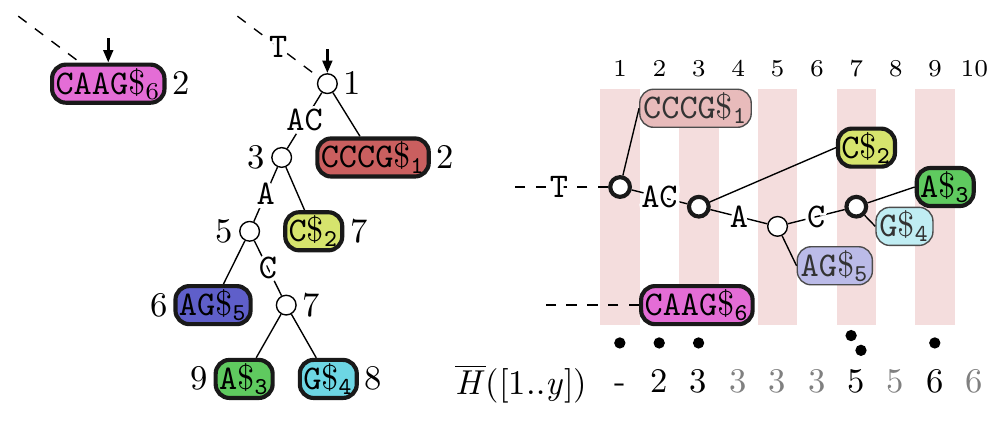}
\caption{On the left, the forest $F_0$ of the example \msa of \Cref{fig:gstmsaheight}, annotated with values $\pos(v)$. On the right, the same forest is plotted against the \msa columns, with only the non-first-born nodes of $\widehat{F}_0$ highlighted. Note that $f(0) = f^6(0) = 2$.}\label{fig:pahexample}
\end{figure}

\begin{definition}[First-born nodes]
Given \gst and its forest $F_x$ corresponding to all semi-repeat-free strings starting from column $x+1$, with $0 \le x < n$, for each internal node $v \in F_x$ we arbitrarily choose one of its children with minimum value of $\pos$ to be a \emph{first-born node} of $F_x$.
Then, let $\widehat{F}_x$ be the subset of non-first-born nodes of $F_x$, obtained by removing these nodes from $F_x$.
\end{definition}

\begin{lemma}
Given \gst and the set of non-first-born nodes $\widehat{F}_x$ associated with column $x$, for any $y \in [f(x)..n]$ the prefix-aware height of segment $[x+1..y]$ is equal to the number of nodes in $\widehat{F}_x$ having $\pos$ value equal or smaller than $y$, in symbols
$
    \pah([x+1..y]) =
    \big\lvert \lbrace
    \widehat{v} \in \widehat{F}_x : \pos(\widehat{v}) \le y
    \rbrace \big\rvert.
$
\end{lemma}
\begin{proof}
For any $v \in F_x$, let $I_v$ be the set of row indexes $i$ such that $v$ covers the leaf of \gst corresponding to $\spell(\msaij{i}{x+1..n})$.
From the properties of \gst it follows that if any $v_1,v_2 \in F_x$ are not one ancestor of the other,
then the corresponding strings $\sstring(\parent(v_1)) \cdot \cchar(v_1)$ and $\sstring(\parent(v_2)) \cdot \cchar(v_2)$ are not one prefix of the other: if $y \ge \pos(v_1)$ and $y \ge \pos(v_2)$, then $\pah(\msaij{I_{v_1} \cup I_{v_2}}{x+1..y}) = \pah (\msaij{I_{v_1}}{x+1..y}) + \pah (\msaij{I_{v_2}}{x+1..y})$. We call this key property the \emph{independence of collateral relatives}.\footnote{In genealogical terms, the ancestor relationship is described as a direct line, as opposed to a collateral line for relatives that are not in a direct line.} In particular, the property holds for any subset $U$ of children of some node $v \in F_x$, provided $y \ge \max_{u \in U} \pos(u)$, and it holds for the exclusive ancestors $W_x \subseteq F_x$, because $y \ge f(x) \ge f^i(x) > \max_{w \in W_x} \pos(w)$ by assumption:
\begin{equation}\label{eq:posindependence}
    \overline{H} \big(
    \msaij{1..m}{x+1..y}
    \big) =
    \sum_{w \in W_x} \overline{H} \big( \msaij{I_w}{x+1..y} \big).
\end{equation}

We can now prove the modification of the thesis restricted to the rows $I_v$ of any node $v \in F_x$.
To do so, we introduce one final notation: we denote with $\widehat{F}_x^v$ the set $(\widehat{F}_x \cap \mgst(v)) \cup \lbrace v \rbrace$, that also deals with the case when $v$ is a first-born node.
Then, for any $v \in F_x$ and $y \in [\max_{i \in I_v}f^i(x)..n]$ we have that
\begin{equation} \label{eq:posinduction}
\overline{H} \big(
    \msaij{I_v}{x+1..y}
    \big) =
    \begin{cases}
        1 & \text{if}\;y < \pos(v),\\
        \big\lvert \big\lbrace
        \hat{v} \in \widehat{F}_x^v : \pos(\hat{v}) \le y
        \big\rbrace \big\rvert
        & \text{otherwise.}
    \end{cases}
\end{equation}
The proof of \Cref{eq:posinduction} proceeds by induction on the height of the subtree rooted at $v$.
\begin{description}
\item[Base case:] If $v$ is a leaf then $\widehat{F}_j^v = \lbrace v \rbrace$, $I_v = \lbrace i \rbrace$ for some $i \in [1..m]$, and $\pah(\msaij{\lbrace i \rbrace}{x+1..y}) = 1$, so \Cref{eq:posinduction} is easily verified.
\item[Inductive hypothesis:] \Cref{eq:posinduction} holds for all nodes $v$ such that the subtree rooted at $v$ has height less than or equal to $h \ge 0$.
\item[Inductive step:] Let the height of the subtree rooted at $v$ be equal to $h+1$, and let $u_1, \dots, u_p$ be the $p \ge 2$ children of $v$, with $u_1$ the first-born.
If $y < \pos(v)$ then all occurrences of $Sc = \sstring(\parent(v)) \cdot \cchar(v)$ in the \msa end after column $y$, so all strings $\spell(\msaij{i}{x+1..y})$ with $i \in I_v$ are prefixes of $Sc$ and $\pah\big( \msaij{I_v}{x+1..y} \big) = 1$.
Using the same argument, \Cref{eq:posinduction} is also verified if $y < \pos(u_1)$, so we can assume $y \ge \pos(u_1) > \pos(v)$.
Consider the children $u_k$ of $v$ such that $y < \pos(u_k)$, for $2 \le k \le p$; the strings spelled in the corresponding rows $I_{u_k}$ are prefixes of $\sstring(\parent(u_1))$, so they are ignored in the prefix-aware height. If $U_{\le} \coloneqq \lbrace u_k : 1 \le k \le p \wedge \pos(u_k) \le y \rbrace$ then
\begin{align*}
    \pah \big( \msaij{I_v}{x+1..y} \big)
    &= \pah \Bigg( \textsf{MSA} \Bigg[ \bigcup_{u \in U_{\le}} I_{u} \Bigg] \Bigg[ x+1..y \Bigg] \Bigg) \\
    &= \sum_{u \in U_{\le}} \pah \big(
        \msaij{I_{u}}{x+1..y} \big) &\substack{\text{\small independence}\\\text{\small of collateral}\\\text{\small relatives}} \\
    &= \sum_{u \in U_{\le}} \big\lvert \big\lbrace \widehat{u} \in \widehat{F}_x^u : \pos(\widehat{u}) \le y \big\rbrace \big\rvert & \substack{\text{\small inductive}\\\text{\small hypothesis}} \\
    &= \big\lvert \big\lbrace
        \hat{v} \in \widehat{F}_x^v : \pos(\hat{v}) \le y
        \big\rbrace \big\rvert.
\end{align*}
Note that the last equality holds because $\pos(u_1)$ of $\widehat{F}_x^{u_1}$ is replaced by $\pos(v)$ of $\widehat{F}_x^v$.
\end{description}
The thesis follows from \Cref{eq:posindependence,eq:posinduction}, because the exclusive ancestors partition the rows $[1..m]$ into $\lvert W_x \rvert$ sets.
Also, note that $\lvert \widehat{F}_x \rvert = m$.
\end{proof}

An example of sets $F_x$ and $\widehat{F}_x$ can be seen in \Cref{fig:pahexample}.
Unfortunately, their naive computation takes time $O(m^2)$ if done locally, because \gst does not contain the information on the ending positions of \msa substrings---and it cannot be easily augmented to do so.
\begin{lemma}\label{lem:mpare1}
Given a general $\msaij{1..m}{1..n}$, \gst, and the exclusive ancestors $W_x$, we can compute the meaningful prefix-aware right extensions $\overline{R}_x$ in time $O(m^2)$.
\end{lemma}
\begin{proof}
For each $v \in F_x$, we can compute $\pos(v)$ by finding for each row $i \in I_v$ the ending position $y_i$ of the occurrence of $Sc = \sstring(\parent(v))\cdot\cchar(v)$ in $\msaij{i}{1..n}$ ($Sc$ is a semi-repeat-free substring so there is at most one occurrence per row).
In other words, position $y_i$ correponds to the $k$-th non-gap character of row $i$, where $k = \rrank( \msaij{i}{1..n}, x ) + \stringdepth(\parent(v)) + 1$.
Then, $\pos(v) = \min_{i \in I_v} y_i$.
The first-born child of $v$ can be found by choosing one of its children with minimum $\pos$ values, and the removal of first-born nodes results in the $\pos$ values of $\widehat{F}_x$.
Given $f(x)$ and the ordered $\pos$ values, a simple algorithm like \Cref{alg:mpahext} considers all columns containing $\pos$ values and outputs the relative prefix-aware right extension as a triple $(x,y,\pah([x+1..y]))$.

\begin{algorithm}
\caption{Algorithm computing the meaningful prefix-aware right extensions, given the sorted $\mathrm{pos}$ values of $\hat{F}_x$. Set $\hat{F}_x$ is obtained by removing the first-born nodes from $F_x$.
}\label{alg:mpahext}
\SetKw{KwOutput}{output}
\KwIn{Value $f(x)$, values $\mathrm{pos}(w_1), \dots, \mathrm{pos}(w_m)$ of nodes in $\hat{F}_x$, sorted from smallest to largest order.}
\KwOut{Meaningful prefix-aware right extensions $\overline{R}_x$.}
$h \gets 1$\;
\While{$\mathrm{pos}(w_h) < f(x)$}{
    $h \gets h + 1$\;
}
\While{$h \le m$}{
    \If{$h = m \;\vee\; \mathrm{pos}(w_{h+1}) \neq \mathrm{pos}(w_h)$}{
        $\KwOutput \; (x,\mathrm{pos}(w_h),h)$\;
    }
    $h \gets h + 1$\;
}
\end{algorithm}

Since we can preprocess in linear time the \msa rows to answer rank and select queries in constant time, the computation of each $\pos(v)$ takes $O(\lvert I_v \rvert)$ time.
Forest $F_x$ is composed of compacted trees with $m$ total leaves, so it contains $O(m)$ nodes: the subtrees of $F_x$ can be unbalanced, hence the total time is $O(m^2)$.
Then, these values can be sorted in time $O(m \log m)$ and then processed in $O(m)$ time.
\end{proof}

\subsection{Speedup using the generalized prefix tree}\label{sub:pos}
Thanks to \Cref{lem:mpare1}, we can compute the meaningful prefix-aware right extensions in $O(m^2n)$ time, the bottleneck being the computation of values $\pos(v)$ (\Cref{def:pos}).
The other tasks can be executed in time $O(mn)$ by adapting the solution to a global computation: the $\pos$ values of all $\widehat{F}_0$, \dots, $\widehat{F}_{n-1}$ are $O(mn)$ in total; all together, they can be sorted in $O(mn)$ time since they take value in $[1..n]$, and they can be separately processed again in total linear time.
\gst does not contain the information about the ending positions of \msa strings, but it does contain the information on the first starting positions (i.e.\ the occurrences): indeed, the sets of leaves $\mathcal{L}_0$, \dots, $\mathcal{L}_{n-1}$ consider each and every suffix starting from a certain \msa column.
This gives us the key intuition of exploiting symmetry to compute values $\pos(v)$ efficiently.

\begin{definition}\label{def:gpt}
Given a general $\msaij{1..m}{1..n}$, we define \gpt as the generalized prefix tree of the set of strings
$\lbrace \$_i \cdot \spell( \msaij{i}{1..n} ) : 1 \le i \le m \rbrace$, with $\$_1, \dots, \$_m$ $m$ new distinct terminator symbols not in $\Sigma$.
Alternatively, \gpt can be constructed as the generalized suffix tree of
$\lbrace \spell( \msaij{i}{1..n} )^{-1} \cdot \$_i : 1 \le i \le m \rbrace$.
\end{definition}

\begin{figure}[htp]
\centering
\includegraphics{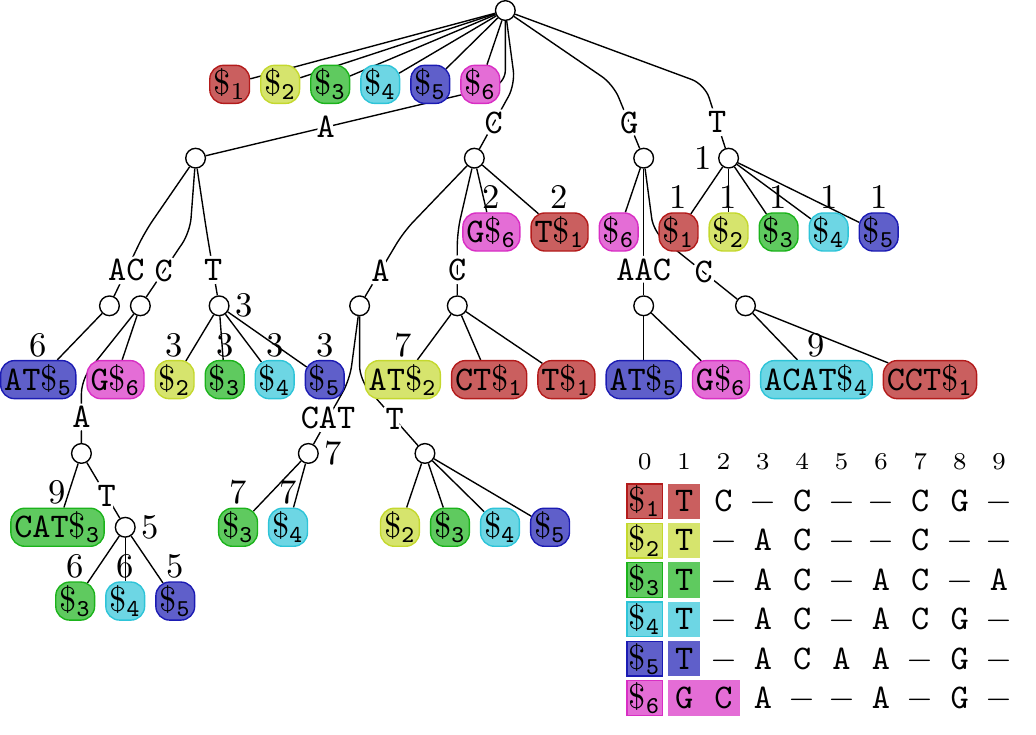}
\caption{Example of the \gpt built from the \msa of \Cref{fig:gstmsaheight}, annotated with the $\pos^{-1}$ values relevant for the computation of $R_0$ (\Cref{fig:pahexample}), the meaningful right extensions starting from column $1$.
}\label{fig:gptmsa}
\end{figure}

\begin{observation}\label{obs:inversepos}
Note that in \gpt strings are read from right to left. For each node $u$ of \gpt, let $\pos^{-1}(u)$ be the first ending position of $\sstring(u)$ in some \msa row:
\begin{itemize}
    \item for any leaf $\ell$ of \gpt corresponding to row $i \in [1..m]$, we have that $\pos^{-1}(\ell)$ is equal to the $k$-th non-gap character of $\msaij{i}{1..n}$, with $k = \lvert \sstring(\ell) \rvert$;
    \item for any internal node $u$ of \gpt, let $v_1$, \dots, $v_p$ be its children; then $\pos^{-1}(u) = \min_{k = 1}^{p} \pos^{-1}(v_k)$;
    \item given a node $v$ of \gst, let $v^{-1}$ be the node of \gpt corresponding to string $\sstring(\parent(v))\cdot \cchar(v)$ read from right to left; if this is an implicit node, then we define $v^{-1}$ as the first explicit ancestor in \gpt; then $\pos(v) = \pos^{-1}(v^{-1})$.
\end{itemize}
\end{observation}
For example, if $v$ is the \gst node of \Cref{fig:gstmsaheight,fig:pahexample} corresponding to string $\mathtt{TAC}\cdot\mathtt{A}$, $v^{-1}$ corresponds to $\mathtt{ACAT}$ in the \gpt of \Cref{fig:gptmsa} and $\pos^{-1}(v^{-1}) = 5$.
\begin{lemma}\label{lem:pos}
Given a general $\msaij{1..m}{1..n}$, \gst, and \gpt, values $\pos(v)$ for any node $v$ of \gst can be computed in $O(mn)$ time.
\end{lemma}
\begin{proof}
As shown in \Cref{obs:inversepos}, the tree structure of \gpt makes it possible to compute $\pos^{-1}(v)$ recursively: similarly to the computation of $\mathcal{L}_0$, \dots, $\mathcal{L}_{n-1}$ in \gst, this takes $O(mn)$ time.
It remains to show that given node $v$ of $\mgst$ we can find node $v^{-1}$ of $\mgpt$ in $O(1)$ time: a data structure representing a synchronized suffix and prefix tree of a string has been developed in the form of the affix tree, or of the corresponding affix array), admitting linear-time construction algorithms~\cite{DBLP:journals/algorithmica/Maass03,DBLP:journals/tcs/Strothmann07}.
However, these results hold under the assumption that the alphabet $\Sigma$ has constant size.
Alternatively, it is straightforward to locate just one occurrence of $\sstring(\parent(v)) \cdot \cchar(v)$ in the \msa, so we can find $v^{-1}$ by answering the corresponding weighted ancestor query in \gpt.
Belazzougui et al.\ recently proved that we can preprocess suffix trees in linear time to be able to answer weighted ancestor queries in constant time, with no assumption about the size of the alphabet~\cite{DBLP:conf/cpm/BelazzouguiKPR21}.
\end{proof}
This concludes the proof of \Cref{theo:pah}:
as we have already shown in \Cref{sub:complexity}, the meaningful prefix-aware right extensions are a drop-in replacement for \Cref{alg:mainalg} of the original meaningful right extensions, so \Cref{lem:pos} implies that the optimal segmentation minimizing the maximum prefix-aware height can be computed in linear time.

\section{Indexing a set of predefined paths in a semi-repeat-free \efg}\label{sect:indexing-efg-with-paths}
The \efg, as a representation of the original \msa sequences, is a lossy data structure: the graph spells the original sequences but also their recombination, with no immediate mechanism deciding whether a pattern actually occurs in the original sequences.
As an enhancement of the \efg framework towards practical applications, we concentrate on a BWT-based index to find the exact matches of a given pattern in a semi-repeat-free \efg, while considering only a predefined set of paths in the graph.

\begin{problem}[Indexing \efgs for the path listing problem]\label{problem:pathlisting}
    Given a semi-repeat-free \efg $G = (V,E,\ell)$ and a set $\mathcal{P}$ of paths in $G$, index $G$ to answer the following query: given a pattern $Q$, find and report paths $P$ of $\mathcal{P}$ containing pattern $Q$ as a substring.
\end{problem}

Consequently, the \efg may be used to identify which of the \msa sequences, from which the \efg has been generated, contains a given pattern.
The overall strategy in achieving this is constructing a pair of indexable texts similarly to what was done in Section~\ref{sect:index}, generating a \emph{BWT index} and augmenting it with additional data structures. To this end, we adapt some of the ideas presented by Norri \cite{Norri2022FounderSegmentations}.

Recall the definition of generalized suffix array and $\bwt_T$ of a text $T$ from \Cref{sect:definitions}.

\begin{definition}
	A \emph{BWT index} is a data structure generated from text $T$ that requires $O(|T|s_\mathsf{LF})$ bits of space and supports the following operations.
	\begin{itemize}
		\item Accessing the $i$-th character in $\mbwt_T$ in constant time.
		\item Given a lexicographic range $[l..r]$, \emph{backward searching} character $c$ can be done in $O(t_\mathsf{LF})$ time. The operation results in the lexicographic range of the suffixes of $T$ in $[l..r]$ that are preceded by $c$.
	\end{itemize}
\end{definition}

Throughout this section, $t_\mathsf{LF}$ and $s_\mathsf{LF}$ are variables the values of which depend on the particular BWT index.

\begin{lemma}
	\label{lem:bwt-index-bckm20}
	\cite[Theorem~6.2]{BCKM20}
	There is a BWT index such that $t_\mathsf{LF} = \log \log \sigma + k$ for any positive $k$, assuming that $\sigma \le |T|$ and $\ceil{\log \sigma} \le \mu$ where $\mu$ is the machine word size in bits. The index data structure can be built in $O(|T|)$ time and $O(|T| \log \sigma)$ bits of working space, and requires $|T| \log \sigma(1 + \frac{1}{k}) + O(|T| \log \log \sigma)$ bits of space. By choosing $k = 1$, we have $t_\mathsf{LF} = \log \log \sigma$ and $s_\mathsf{LF} = \log \sigma$.
\end{lemma}

The search algorithm is based on the property of semi-repeat-free \efgs that certain prefixes of the concatenations of the node labels of any two connected nodes are distinct, as shown in the following lemma.

\begin{lemma}
	\label{lem:semi-repeat-free-prefix-property}
	\cite[Lemma~4.4]{Norri2022FounderSegmentations}
	Suppose $G = (V, E, \ell)$ is a semi-repeat-free \efg and $v_1, v_2, w_2 \in V$ are nodes such that $(v_1, w_2) \in E$, $v_2$ is in the same block as $w_2$, and $\ell(v_2) \prefix \ell(w_2)$. There is no path $W = u_1 \dotsm u_t$ of $G$ such that $u_1 \neq v_1$ and $\ell(v_1)\ell(v_2) \prefix \ell(u_1)\dotsm\ell(u_t)$.
\end{lemma}

\begin{table}
	\centering
	\begin{tabular}{llllll}
		\toprule
		\multirow{2}{*}{Cases}		& \multirow{2}{*}{Condition}                                        & \multirow{2}{*}{Node}	& \multicolumn{3}{c}{Node labels}\\
		\cmidrule(lr){4-6}
                                    &                                                                   &						& $i = 1$				& $i = 2$				& $i = 3, \ldots, t$\\
		\midrule
		1,	                        & $\ell(v_1)\ell(v_2) \prefix \ell(u_1)\ell(u_2)$                   & $u_i$					& $\ell(v_1)\cdot Q$	& $\ell(u_2)$			&\\
		2							& $\ell(u_1)\ell(u_2) \prefix \ell(v_1)\ell(v_2)$                   & $v_i$					& $\ell(v_1)$			& $\ell(v_2)$			&\\
		\cmidrule(lr){3-6}
		\multirow{2}{*}{3}			& \multirow{2}{*}{$\ell(v_1)\ell(v_2) \prefix \ell(u_1)\ell(u_2)$}  & $u_i$					& $\ell(u_1)$			& $\ell(u_2)$			&\\
									&                                                                   & $v_i$					& $\ell(u_1)\cdot Q$	& $\ell(v_2)$			&\\
		\cmidrule(lr){3-6}
		\multirow{2}{*}{4}			& \multirow{2}{*}{$\ell(u_1)\ell(u_2) \prefix \ell(v_1)\ell(v_2)$}  & $u_i$					& $\ell(u_1)$			& $\ell(u_2)$			&\\
									&                                                                   & $v_i$					& $\ell(u_1)\cdot Q$	& $\ell(v_2)$	&\\
		\cmidrule(lr){3-6}
		\multirow{2}{*}{4.1}        & \multirow{2}{*}{$S$ may be $\varepsilon$}                         & $u_i$                 & $\ell(u_1)$			& $Q \cdot R \cdot S$	&$T$\\
		                            &                                                                   & $v_i$                 & $\ell(u_1)\cdot Q$    & $Q \cdot R$           &\\
		\cmidrule(lr){3-6}
		\multirow{2}{*}{4.2}        & \multirow{2}{*}{$S\ne\varepsilon$}                                & $u_i$                 & $\ell(u_1)$			& $Q \cdot R$			& $S\cdot T$\\
		                            &                                                                   & $v_i$                 & $\ell(u_1)\cdot Q$    & $Q \cdot R \cdot S$   &\\
		\bottomrule
	\end{tabular}
	\caption{\label{tab:semi-repeat-free-prefix-property-node-labels-and-decomposition}Node labels and decompositions in Lemma~\ref{lem:semi-repeat-free-prefix-property}, given that $W$ consists of at least two nodes.}
\end{table}

\begin{proof}
	We show that if $W$ existed, $G$ would not be semi-repeat-free. We note that if $\ell(v_1)\ell(v_2) \prefix \ell(u_1)\dotsm\ell(u_t)$, then $(\ell(v_1)\ell(v_2))[1, k] = (\ell(u_1)\dotsm\ell(u_t))[1, k]$ for any $k \le |\ell(v_1)\ell(v_2)|$ and vice-versa.
	
	Suppose that $W$ consists of only one node, $u_1$. Then $\ell(v_2)$ is a non-prefix substring of $\ell(u_1)$, a contradiction since $G$ is semi-repeat-free.

	Suppose now that $W$ consists of at least two nodes. There are four possibilities for decomposing the node labels, as summarised in Table~\ref{tab:semi-repeat-free-prefix-property-node-labels-and-decomposition}, where $Q$ and $S$ are some non-empty strings:
\begin{enumerate}
 \item $\ell(u_1) = \ell(v_1) \cdot Q$ and $\ell(v_1)\ell(v_2) \prefix \ell(u_1)\ell(u_2)$.
 If $W$ existed, $\ell(v_2)$ would be a prefix of $Q \cdot \ell(u_2)$ and consequently a non-prefix substring of $\ell(u_1)\ell(u_2)$, a contradiction since $G$ is semi-repeat-free.

 \item $\ell(u_1) = \ell(v_1) \cdot Q$ and $\ell(u_1)\ell(u_2) \prefix \ell(v_1)\ell(v_2)$.
 It follows that $Q\cdot\ell(u_2)$ is a prefix of $\ell(v_2)$ and consequently $\ell(u_2)$ is a non-prefix substring of $\ell(v_2)$, again a contradiction since $G$ is semi-repeat-free.

 \item $\ell(v_1) = \ell(u_1) \cdot Q$ and $\ell(v_1)\ell(v_2) \prefix \ell(u_1)\ell(u_2)$.
 Similarly to the previous case, $Q\cdot\ell(v_2) \prefix \ell(u_2)$ and consequently $\ell(v_2)$ is a non-prefix substring of $\ell(u_2)$, a contradiction.
 \item $\ell(v_1) = \ell(u_1) \cdot Q$ and $\ell(u_1)\ell(u_2) \prefix \ell(v_1)\ell(v_2)$.
 Since $G$ is semi-repeat-free and $Q$ is a suffix of $\ell(v_1)$, $Q$ also has to be a prefix of both $\ell(u_2)$ and $\ell(v_2)$. Furthermore, either $\ell(u_2) = Q\cdot R$ or $\ell(v_2) = Q\cdot R$ where $R$ is a non-empty string since otherwise $Q$ as the node label of either node would appear as a proper suffix of $\ell(v_1)$ and consequently $G$ would not be semi-repeat-free. We now have two possibilities for decomposing $\ell(u_2)$ and $\ell(v_2)$:
 \begin{enumerate}
     \item $\ell(u_2) = Q\cdot R\cdot S$ and $\ell(v_2) = Q\cdot R$.
     If $S = \varepsilon$, $u_2 = v_2$. Since $\ell(v_1)\ell(v_2) = \ell(u_1)\cdot Q\cdot Q\cdot R$, it must be that $Q\cdot Q\cdot R\cdot S \prefix \ell(u_2)\dotsm\ell(u_t)$. It now holds that $\ell(v_2) = Q\cdot R$ is a non-prefix substring of $\ell(u_2)\dotsm\ell(u_t)$:
	\[
	    \begin{array}{rl}
	        Q\cdot \underbrace{Q\cdot R}_{\ell(v_2)} &\prefix \; \underbrace{Q\cdot R\cdot S}_{\ell(u_2)} \mathrlap{\underbrace{\phantom{Q}}_{\ell(u_3)\dotsm\ell(u_t)}} \; \cdot \;\;\; T.\\
	    \end{array}
	\]
	Hence $G$ would not be semi-repeat-free.

    \item $\ell(u_2) = Q\cdot R$ and $\ell(v_2) = Q\cdot R\cdot S$ for some $S \ne \varepsilon$.
	We now have
	
	\[
    	\begin{array}{rl}
    		Q\cdot \underbrace{Q\cdot R\cdot S}_{\ell(v_2)}	&\prefix\quad \underbrace{Q\cdot R}_{\ell(u_2)} \mathrlap{\underbrace{\phantom{Q\cdot R}}_{\ell(u_3)\dotsm\ell(u_t)}} \, \cdot \; S\cdot T.
    	\end{array}
	\]\\
	Since $Q\cdot Q\cdot R\cdot S$ = $Q\cdot\ell(v_2)$, $Q\cdot R = \ell(u_2)$, $S\cdot T = \ell(u_3)\dotsm\ell(u_t)$, and none of $Q, R, S$ and $T$ equals $\varepsilon$, $\ell(v_2)$ is now a non-prefix substring of $\ell(u_2)\dotsm\ell(u_t)$, a contradiction, since $G$ is semi-repeat-free.\qedhere
 \end{enumerate}
\end{enumerate}
\end{proof}

By using this result, we can construct an indexable text analogous to $D_F$ and $D_R$ from \Cref{sect:index} and consisting of the concatenation of the strings spelled by the edges of the \efg.
In order to have each node label at a particular position in the lexicographic order of the suffixes of the text, we also add the labels of the nodes in the last \efg block, that is, nodes that only have in-edges.

\begin{align*}
	D'_F	&= \0 \cdot \prod_{v \in V^{-}} \ell(v)\0		\cdot \prod_{(v, w) \in E} \ell(v)\ell(w)\0		&\textrm{where\;} V^{-} = \{v \in V \mid (v, \cdot) \notin E\}
\end{align*}
Determining if a given pattern occurs as a substring of some path in an \efg can then be done as follows.

\begin{lemma}
	\label{lem:semi-repeat-free-text-representation-lexicographic-range-nestedness}
	Suppose $G = (V, E, \ell)$ is a semi-repeat-free \efg, and $[l..r]$ and $[l'..r']$ are the lexicographic ranges of some $\ell(v)$ and $\ell(w)$ in $\sa_{D'_F}$ (not necessarily in the same order) where $v, w \in V$. Then $l' \le l$ and $r \le r'$ if and only if $v$ and $w$ are in the same block of $G$.
\end{lemma}

\begin{proof}
	We prove the Lemma by showing that each of the conditions implies the other.
	
	($\Rightarrow$) Since the lexicographic ranges are nested, either $\ell(v)$ must be a prefix of $\ell(w)$ or vice-versa. Since $G$ is semi-repeat-free, $v$ and $w$ must be in the same block.
	
	($\Leftarrow$) Since $v$ and $w$ are in the same block of $G$ and $\ell(v) \prefix \ell(w)$ or vice-versa, their lexicographic ranges in $\sa_{D'_F}$ must be nested.
\end{proof}

\begin{lemma}[Expanded backward search]
	\label{lem:expanded-backward-search-tentative}
	\cite[Lemma~4.6]{Norri2022FounderSegmentations}
	Suppose that $G = (V, E, \ell)$ where $\ell \colon V \to \{1, \ldots, \sigma\}$ is a non-trivial semi-repeat-free \efg. By preprocessing $G$ in $O(L(|E| + |V|t_\mathsf{LF}))$ time to generate a data structure of $O(L|E|s_\mathsf{LF})$ bits, we can determine for any pattern $Q$ an encoding of nodes $\{v_1, v_2, \ldots, v_k\} \subset V$ such that for all $1 < i \le k$, there is $(v_{i - 1}, w_i) \in E$ with $w_i$ in the same block as $v_i$, $\ell(w_i) \prefix \ell(v_i)$, and $Q$ a substring of $\ell(v_1)\ell(v_2)\dotsm\ell(v_k)$. Here $L = \max_{(v, w) \in E} ( \lvert \ell(v) \rvert+ \lvert \ell(w) \rvert )$. The time requirement of the operation is $O(|Q|t_\mathsf{LF})$, if $Q$ spans the concatenation of at least two node labels in $G$.
    Otherwise, determining all edge labels containing $Q$ takes $O((|Q| + L \cdot \mathrm{occ})\cdot t_\mathsf{LF})$ time where 
    $\mathrm{occ}$ is the number of occurrences of $Q$ in $D'_F$.
\end{lemma}

\begin{proof}
	First, we construct a BWT index with $D'_F$ as input. This can be done in $O(|D'_F|) = O(L(|E| + |V|)) = O(L|E|)$ time.
	If $Q$ is a substring of any $\ell(u)\ell(v)$ where $(u, v) \in E$, utilising regular backward search in $O(|Q|t_\mathsf{LF})$ time suffices.
	
	To make it possible to handle patterns that span more nodes, as part of preprocessing we also fill two bit vectors, $\mathcal{B}$ and $\mathcal{E}$, of $|D_F|$ elements each. We set $\mathcal{B}[l] = 1$ and $\mathcal{E}[r] = 1$ if and only if $[l..r]$ is the lexicographic range of some $\ell(v)$, $v \in V$, such that $v$ is the shortest prefix of any node label among all node labels of $G$. Formally, the set of such nodes is $B := \{\argmin_v |\ell(v)| \mid \ell(v) \prefix \ell(w) \textrm{ and } v, w \in V\}$. Since none of the labels of the nodes in $B$ is a substring of another, the corresponding lexicographic ranges are disjoint. Additionally $\mathcal{B}$ is prepared for $\mathrm{rank}$ queries and both $\mathcal{B}$ and $\mathcal{E}$ are prepared for $\mathrm{select}$ queries. The preprocessing can be done in $O(L|V|t_\mathsf{LF})$ time.
	
The bit vectors $\mathcal{B}$ and $\mathcal{E}$ are utilized as a part of searching as follows. Suppose $Q$ has been processed such that the resulting lexicographic range is $[l..r]$ and $\bwt_{D'_F}[l] = \0$.
 We then \emph{expand the lexicographic range} by determining a pair of indices $l'$ and $r'$ such that $\mathcal{B}[l'] = \mathcal{E}[r'] = 1$, $l' \le l$ and $r \le r'$, and $l'$ and $r'$ correspond to the same lexicographic range. This can be done by first calculating $r'' := \mathrm{rank}(\mathcal{B}, l)$, $l' = \mathrm{select}(\mathcal{B}, r'')$ and $r' = \mathrm{select}(\mathcal{E}, r'')$, and then checking if in fact $l' \le l$ and $r \le r'$. Suppose $v$ and $w$ are the nodes of $G$ where $[l..r]$ is the lexicographic range of $\ell(v)$ and $[l'..r']$ is that of $\ell(w)$. By Lemma~\ref{lem:semi-repeat-free-text-representation-lexicographic-range-nestedness} there must be a node $w' \in V$ such that $(w', v) \in E$ and $\ell(w) \prefix \ell(w')$.

    When the backward search is continued and it is determined from $\bwt_{D'_F}$ that the preceding character is again $\0$, by Lemma~\ref{lem:semi-repeat-free-prefix-property} the current node can be determined unambiguously from the lexicographic range. The search can then be resumed from another node in the same block by expanding the lexicographic range again. If the pattern is short in the sense that the lexicographic range is never expanded, the conventional method of branching and continuing the search with the preceding characters as determined from $\bwt_{D'_F}$ may be applied until a $\0$ has been found in every branch, reaching the stated time complexity of $O((\lvert Q \rvert + L \cdot \mathrm{occ}) \cdot t_{\mathsf{LF}})$.
\end{proof}

In order to filter out both the false positives that do not occur in the \efg as well as those not in the text from which the \efg was constructed, we first prove a property of the \efg and then show how to make use of some additional data structures.

\begin{lemma}
    \label{lem:semi-repeat-free-edge-label-lex-range-singleton}
    Suppose $G = (V, E, \ell)$ is a semi-repeat-free \efg and $(v, w) \in E$. The (co-)lexicographic range of $\ell(v)\ell(w)\mathbf{0}$ in $\sa_{D'_F}$ is a singleton.
\end{lemma}

\begin{proof}
    Since $G$ is semi-repeat-free and $\mathbf{0}$ is not in the node label alphabet, $\ell(v)\ell(w)\mathbf{0}$ cannot occur elsewhere in $D'_F$, in particular as a suffix or infix of another node label concatenation.
\end{proof}

\begin{lemma}
    \label{lem:semi-repeat-free-edge-numbering}
    \cite[Lemma~4.8]{Norri2022FounderSegmentations}
    Suppose $G = (V, E, \ell)$ is a non-trivial semi-repeat-free founder graph. The edges $E$ can be assigned distinct numbers with the mappings $\alpha \colon E \to \{1, \ldots, |E|\}$ and $\tilde{\alpha} \colon E \to \{1, \ldots, M_E\}$ where $M_E = 2 + H' + |E|$ and $H'$ is the height of the first block of $G$.
    Given the co-lexicographic rank of any $\ell(v)\ell(w)\0$ in $\sa_{D'_F\$}$ where $(v, w) \in E$, the value of $\tilde{\alpha}((v, w))$ can be determined in constant time. Similarly, given the lexicographic rank of the aforementioned node label concatenation, the value of $\alpha((v, w))$ can be retrieved in constant time by utilizing a data structure of $O(|D'_F|)$ bits that can be prepared in $O(L|E|t_\mathsf{LF})$ time where $L = \max_{(v, w) \in E} |\ell(v)| + |\ell(w)|$.
\end{lemma}

\begin{proof}
    Since $\$$ and $\0$ are the lexicographically smallest characters in the alphabet of $\bwt_{D'_F}$, an upper bound of the co-lexicographic ranks in question in $\sa_{D'_F\$}$ is $2 + H' + |E| = M_E$, the first two smaller ranks being those of $D'_F\$$ and the initial $\0$. Since the segmentation is non-trivial, there must be at least as many edges in $E$ as there are nodes in the first block, and thus $M_E = O(|E|)$.

    For calculating the values of $\alpha$ we prepare bit vector $\mathcal{D}$ of $D'_F$ elements. We set $\mathcal{D}[l] = 1$ iff. $l$ is the lexicographic rank (see Lemma~\ref{lem:semi-repeat-free-edge-label-lex-range-singleton}) of some $\ell(v)\ell(w)\0$ in $\sa_{D'_F\$}$ where $(v, w) \in E$. After preparing $\mathcal{D}$ for constant time rank queries, the value of $\alpha((v, w))$ can be retrieved with $\rank(\mathcal{D}, l)$.
\end{proof}

\begin{lemma}
	\label{lem:semi-repeat-free-with-paths-count}
	\cite[Lemma~4.10]{Norri2022FounderSegmentations}\footnote{While the original version \cite[Lemma~4.10]{Norri2022FounderSegmentations} utilized maintaining also the co-lexicographic range as part of the backward search, it was observed that (quite trivially) by processing the pattern $Q$ twice, this is not actually necessary. Consequently, the multiplier in the proven time complexity in Theorem~\ref{thm:semi-repeat-free-with-paths-count} was reduced from $\log \sigma$ to $\log \log \sigma$.}
	Suppose $G = (V, E, \ell)$ is a non-trivial semi-repeat-free \efg of $b$ blocks and $\mathcal{P}$ is a set of paths. The index data structure of Lemma~\ref{lem:expanded-backward-search-tentative} can be augmented for listing paths $P$ containing any given pattern. To this end, preprocessing $G$ can be done in $O(L|E|t_\mathsf{LF} + |V||\mathcal{P}|)$ time to generate an index data structure of $O(|V|(|\mathcal{P}| + \log b) + |E|(Ls_\mathsf{LF} + \log H))$ bits, where $b$ is the number of blocks in $G$, $H$ is the maximum block height, and $L = \max_{(v, w) \in E}| \ell(v)|+ |\ell(w)|$.
	Paths containing $Q$ can then be listed in $O(|Q|t_\mathsf{LF} + (|Q| + H^2)\ceil{\frac{|\mathcal{P}|}{\mu}})$ time if $Q$ spans the concatenation of at least two node labels in $G$, and otherwise 
    $O((|Q| + L\cdot \mathrm{occ})\cdot t_\mathsf{LF} + \mathrm{occ}\ceil{\frac{|\mathcal{P}|}{\mu}})$. Here, $\mu$ is the machine word size in bits and $\mathrm{occ}$ is the number of occurrences of $Q$ in $D'_F$.
\end{lemma}

\begin{proof}
	We would eventually like to have a compressed representation of the \efg. To that end, we define an order for the nodes of $G$ that preserves the block order.
    Such an order that also causes the nodes to be arranged in the lexicographic order of the node labels is $<_\varrho$ defined as follows.
    Suppose $V_i$ and $V_j$ are blocks in $G$ and $v \in V_i$ and $w \in V_j$.
	
	\[
		v <_\varrho w \Leftrightarrow i < j \textrm{ or } (i = j \textrm{ and } \ell(v) < \ell(w))
	\]
	
	We denote the permutation of the nodes to this order by $\varrho \colon V \to \lbrace 1, \dots, \lvert V \rvert \rbrace$. If needed, counting sort \cite{Seward1954Information} may be applied to sort the nodes by block number in $O(|V|)$ time.
    To sort the nodes in the lexicographic order of the their labels, $O(L|V|t_\mathsf{LF})$ time can be spent to construct a BWT index from the concatenation of the node labels separated by special characters and searching for each node label.

	For the compressed representation, we store various properties of the \efg.
	In integer vector $N$, we store the block numbers of the nodes for which a one was stored in $\mathcal{B}$, in the same order, using $O(L|V|t_\mathsf{LF})$ time and $O(|V| \log b)$ bits of space.
	For determining the cumulative sum of the block heights in constant time, we store the height of the blocks as unary in bit vector $B$ and prepare it for $\mathrm{rank}$ and $\mathrm{select}_0$  
	queries. The preparation can be done in $O(|V|)$ time and requires the same amount of space in bits.
	To be able to determine the numbers of the nodes connected by an edge, we also fill four integer vectors, $A$, $\tilde{A}$, $\mathcal{L}$, and $\mathcal{R}$. In $A$ ($\tilde{A}$, resp.), for each $(v, w) \in E$, we store the rank of $v$ ($w$) within its block in the order defined by $<_\varrho$ and using $\alpha((v, w))$ ($\tilde{\alpha}(v, w)$, resp.) as the key. This can be done in $O(L|E|t_\mathsf{LF})$ time and using $O(|E|\log h)$ bits of space. Finally, for each $(v, w) \in E$, we set $\mathcal{L}[\tilde{\alpha}((v, w))] = \mathcal{R}[\alpha((v, w))] = \varrho(w) - \varrho(v)$. This requires $O(|E| \log H)$ bits of space and can be done in $O(L|E|t_\mathsf{LF})$ time, as described in Lemma~\ref{lem:semi-repeat-free-edge-numbering}. To calculate the values of $\tilde{\alpha}$, we build a BWT index of also ${D'_F}^{-1}$.
	
	To be able to determine on which predefined paths a node occurs, we make an inverted index by storing a bit vector for each node. Suppose $f \colon \mathcal{P} \to \lbrace 1, \ldots, \lvert \mathcal{P} \rvert \rbrace$ is a bijection; then we set $U_{\varrho(v)}[f(P)] = 1$ where $U_{\varrho(v)}$ is such a bit vector and $P \in \mathcal{P}$.
	
	As the final preparation step, each $|\ell(v)|$ where $v \in V$ is such a node for which a one was stored in $\mathcal{B}$, is stored as unary in bit vector $X$ in $\varrho$ order. In other words, we set $X[1] = 1$ followed by a run of zeros and a one for each such node. Consequently, after preparing $X$ for constant-time $\rank_0$ and $\select$ queries, 
	the length of the node labels can be determined in constant time. This preparation step can be done in $O(|V|L) = O(|E|L)$ time and requires $O(|V|L)$ bits of space.
	
	We are now ready to determine the subset of shortest paths in $G$ to which the traversed nodes belong as part of the expanded backward search. First, suppose $Q$ spans the label of one node at most. In this case, the usual approach of continuing the search may be applied
    as shown in Lemma~\ref{lem:expanded-backward-search-tentative}. Given some branch, suppose $l$ is its lexicographic rank before locating the first $\0$ character. The block number $b'$ can be determined with $N[\rank(\mathcal{B}, l)]$, and the value of $\varrho(v)$ by first determining $a := \rank(\mathcal{D}, l)$ and then with $\rank(B, \select_0(B, b' + 1)) + A[a]$, each in constant time. The value of $\varrho(w)$ can be calculated by adding $\mathcal{R}[a]$ to $\varrho(v)$, also in constant time. Finally, the set of paths where $Q$ occurs can be be determined with the bitwise operation $U_{\varrho(v)} \wedge U_{\varrho(w)}$.
	
	Now suppose $Q$ spans more than the label of one node. Then there is a decomposition $Q = P \cdot \ell({v_j}_1) \cdot \ell({v_j}_2) \dotsm \ell({v_j}_k) \cdot S$ where ${v_j}_1 \cdot {v_j}_2 \dotsm {v_j}_k$ is a path of at least one node in $G$ and $P$ and $S$ (for prefix and suffix) are strings at least one of which is non-empty, as seen in \Cref{fig:pattern-decomposition}. Note that $S$ contains a full node label as its prefix, whereas $P$ does not contain any full label. When the lexicographic range has been expanded at least twice, as shown in Lemma~\ref{lem:expanded-backward-search-tentative}, the number of the left node of the edge can be determined as follows. Suppose $l$ is the lexicographic rank of the edge $(v, w) \in E$ before the expansion and $[l'..r']$ is the lexicographic range of some node label $\ell(v')$ after the expansion. The block number $b'$ of both $v$ and $v'$ can be determined with $N[\rank(\mathcal{B}, l')]$, and the value of $\varrho(v)$ like in the previous case, both in constant time.

 \begin{figure}
     \centering
     \includegraphics{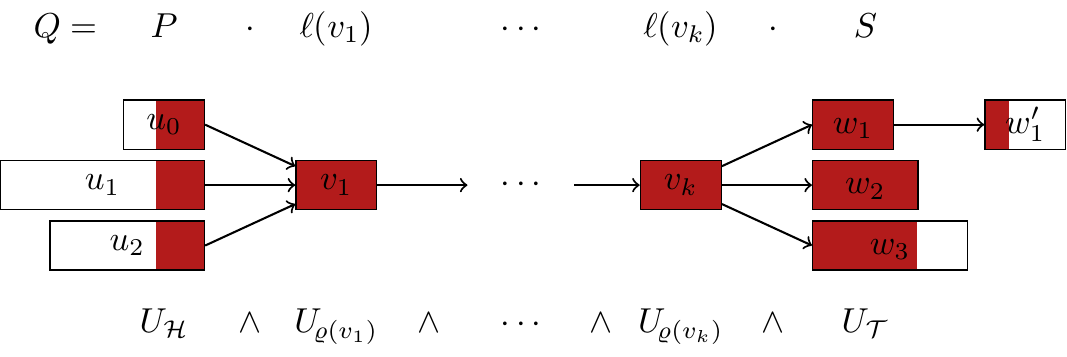}
     \caption{Sub-graph of a semi-repeat-free \efg $G$ showing the different possible occurrences of $Q$ in $G$ decomposed as $P \cdot \ell(u_1) \cdots \ell(u_k) \cdot S$ (marked in blue).
     Below, the bitvectors corresponding to the indexed paths where the corresponding substring of $Q$ appears. Their bitwise AND returns all paths in $\mathcal{P}$ containing $Q$.}
     \label{fig:pattern-decomposition}
 \end{figure}
	
	Considering $P$ and $S$, we would like to determine the nodes in the sets $\mathcal{H} = \{u \mid (u, {v_j}_1) \in E \textrm{ and } P \textrm{ is a suffix of } \ell(u)\ell({v_j}_1)\}$ and $\mathcal{T} = \{(v, w) \mid (v, w) \in E \textrm{ and } S \prefix \ell(v)\ell(w)\}$.

    We denote $H' = P \cdot \ell({v_j}_1)$.
Suppose $[l..r]$ is the co-lexicographic range of $H' \cdot \0$. The values $\tilde{\alpha}(\{(u, {v_j}_1) \mid u \in \mathcal{H}\})$ are now $l, l + 1, \ldots, r$, which can be used as keys for $\mathcal{L}$ to calculate $\varrho(\mathcal{H})$ from $\varrho({v_j}_1)$. There are at most $H$ nodes in $\mathcal{H}$. For $S$, suppose $(u, w) \in \mathcal{T}$ and $r'$ is its lexicographic rank before expansion. The value of $\varrho(u)$ can be determined like before from the expanded lexicographic range and $\varrho(w) = \varrho(u) + \mathcal{R}[r']$. It holds that $|\mathcal{T}| \le H \cdot \min\{H, |Q|\}$.
 The node numbers that correspond to the edges in $\mathcal{T}$ can thus be determined in $O(h \cdot \min\{h, |Q|\})$ time. The predefined paths on which the found paths occur can be determined using bitwise operations with $U_{\mathcal{H}} \wedge U_{\varrho({v_j}_1)} \wedge U_{\varrho({v_j}_2)} \wedge \ldots \wedge U_{\varrho({v_j}_k)} \wedge U_{\mathcal{T}}$ where
	
	\begin{align*}
		U_{\mathcal{H}} &:= \bigvee_{(u, {v_j}_1) \in \mathcal{H}}(U_{\varrho(u)} \wedge U_{\varrho({v_j}_1)}) \textrm{{\quad}and }\\
		U_{\mathcal{T}} &:= \bigvee_{(v, w) \in \mathcal{T}}(U_{\varrho(v)} \wedge U_{\varrho(w)}).
	\end{align*}

    In the edge case where $S \prefix \ell(v)$ for some $v$ that is located in the last block of $G$, instead of using $\mathcal{T}$, we proceed as follows. Suppose $\mathcal{T'} = \{u \in V \mid S \prefix \ell(u) \}$. Also suppose the lexicographic ranges before and after the expansion are $[\tilde{l}..\tilde{r}]$ and $[\tilde{l'}..\tilde{r'}]$ respectively. From $\tilde{l}$ we can determine the number $\varrho(u')$ of the node $u' \in B$ such that $\ell(u') \prefix S$. The numbers of the nodes in $\mathcal{T'}$ are now $\varrho(u') + \tilde{l'} - \tilde{l}$, \ldots, $\varrho(u') + \tilde{r'} - \tilde{l}$, since the nodes are assigned numbers in the lexicographic order of their labels in addition to the block numbers.
    	
	It remains to be shown how to handle the case where the lexicographic range has been expanded exactly once. In this case, the traversed paths do not necessarily converge to one or more nodes, i.e.\ $Q = P \cdot S$ for some non-empty strings $P$ and $S$. However, since the lexicographic range was expanded, $S$ can be decomposed such that $S = \ell(v') \cdot S'$ where $S'$ is some (possibly empty) string. By Lemma~\ref{lem:semi-repeat-free-text-representation-lexicographic-range-nestedness}, for all $(v, w) \in E$ where $S \prefix \ell(v)\ell(w)$, $v'$ must be in the same block as $v$.
	
	Suppose $[l..r]$ and $[l_{v'}..r_{v'}]$ are the lexicographic ranges before and after the expansion, respectively. We can determine $|\ell(v')|$ from $X$ and $|P|$ by counting the searched characters. We can then find all the edges in the set $\mathcal{H'} = \{(u, v) \mid (u, v) \in E \textrm{, } \ell(v') \prefix \ell(v) \textrm{ and some non-empty prefix of } Q \textrm{ is a suffix of } \ell(u)\ell(v)\}$ by forward searching $Q$ one character at a time, testing for $\0$ on each step after $|P| + |\ell(v')|$ characters have been processed. The block number of $v'$ can be determined like before and, for any edge $(u, v) \in \mathcal{H'}$, we can use $\tilde{A}$ and $\mathcal{L}$ to determine both $\varrho(u)$ and $\varrho(v)$. From $[l..r]$ we can determine the values of $\alpha((v, w))$ for all $(v, w) \in \mathcal{T}$, which in turn can be used for calculating the values of $\varrho(w)$ from the values retrieved from $\mathcal{R}$.
	
	The matching paths can be identified as follows. Given a node ${v_j}_0$, we find the subsets $\mathcal{H'}({v_j}_0) = \{(u, {v_j}_0) \mid (u, {v_j}_0) \in \mathcal{H'}\}$ and $\mathcal{T}({v_j}_0) = \{({v_j}_0, w) \mid ({v_j}_0, w) \in \mathcal{T}\}$. We then calculate
	
	\begin{align*}
		U_{\mathcal{H'}({v_j}_0)} &:= \bigvee_{(u, {v_j}_0) \in \mathcal{H'}({v_j}_0)}(U_{\varrho(u)} \wedge U_{\varrho({v_j}_0)}) \textrm{{\quad}and }\\
		U_{\mathcal{T}({v_j}_0)} &:= \bigvee_{({v_j}_0, w) \in \mathcal{T}({v_j}_0)}(U_{\varrho({v_j}_0)} \wedge U_{\varrho(w)}),
	\end{align*}
	and finally $U_{\mathcal{H}'({v_j}_0)} \wedge U_{\mathcal{T}({v_j}_0)}$. Since there are at most $2H^2$ edges in $\mathcal{H'}$ and $\mathcal{T}$ in total, the bitwise operations may be done in $O(H^2 \ceil{\frac{|\mathcal{P}]}{\mu}})$ time by partitioning the edges by the common nodes.
\end{proof}

By combining the result of Lemma~\ref{lem:semi-repeat-free-with-paths-count} with the BWT index of Lemma~\ref{lem:bwt-index-bckm20}, we have the following result.

\begin{theorem}
	\label{thm:semi-repeat-free-with-paths-count}
	\cite[Theorem~4.11]{Norri2022FounderSegmentations}
	Suppose $G = (V, E, \ell)$ is a non-trivial semi-repeat-free \efg of $b$ blocks and $\mathcal{P}$ is a set of paths. Graph $G$ can be preprocessed for
    solving the path listing problem (Problem~\ref{problem:pathlisting})
    in $O(L|E|\log \log \sigma + |V||\mathcal{P}|)$ time to build a data structure of $O(|V|(|\mathcal{P}| + \log b) + |E|(L \log \sigma + \log h))$ bits, assuming that $\sigma \le L|E|$ and $\log \sigma \le \mu$.
    The path listing problem can then be solved, given pattern $Q$,
    in $O(|Q| \log \log \sigma + (|Q| + H^2)\ceil{\frac{|\mathcal{P}|}{\mu}})$ time if $Q$ spans the concatenation of at least two node labels in $G$, and otherwise $O((|Q| + L\cdot \mathrm{occ})\log \log \sigma + \mathrm{occ}\ceil{\frac{|\mathcal{P}|}{\mu}}))$ time. Here, $L = \max_{(v, w) \in E}|\ell(v)|+ |\ell(w)|$, $H$ is the maximum block height, $\mu$ is the machine word size in bits and $\mathrm{occ}$ is the number of occurrences of $Q$ in $D'_F$.
\end{theorem}

\section{Discussion}\label{sect:discussion}
We improved and extended the theory on Elastic Founder Graphs, overcoming the difficulty of building and indexing a graph representing a multiple sequence alignment \emph{with gaps}:
we improved from $O(NH^2)$ to $O(NH)$ the space occupied by the index enabling linear-time pattern matching on \efgs, where $N$ is the total label length and $H$ is the height of the graph (\Cref{theo:index});
by formulating and solving the exclusive ancestor set problem on trees (\Cref{prob:exclusiveancestorset} and \Cref{lem:exclusiveancestorset}),
we improved to linear-time the $O(mn \log m)$ and $O(mn \log m + n \log \log n)$ construction algorithms segmenting \msas while maximizing the number of blocks or minimizing the maximum segment length (\Cref{corol:maxblocks,corol:minmaxlength});
we developed an $O(mn)$-time segmentation algorithm minimizing the maximum block height in the gapless setting (\Cref{theo:gaplessheight}), and extended to the general case the solution to work in parameterized linear-time for constant-sized alphabets, but ultimately in quadratic-time (\Cref{theo:generalheight}); we studied the refined measure of \emph{prefix-aware height} and obtained an $O(mn)$ construction algorithm minimizing the maximum for this height (\Cref{theo:pah}); finally, we showed an extension of \efgs indices to solve path listing queries on a set of indexed \efg paths, and thus on the original \msa sequences (\Cref{problem:pathlisting} and \Cref{thm:semi-repeat-free-with-paths-count}).

The segmentation algorithms we developed (\Cref{alg:minmaxlength,alg:mainalg}) accept in input a range-based description of the valid segments and the corresponding score for each chosen segment.
We believe that they can be applied to any setting concerned with the optimal constrained segmentation of aligned sequences, and they can be the basis for a further generalization to different score criteria.

\efgs naturally compress the chosen semi-repeat-free segments, if they contain aligned substrings in the \msa spelling the same string.
The optimization metrics we studied---the number of blocks, the segment length, and the block height---target a specific feature of the \efg and can be considered a heuristic solution to a compressed graph index: indeed, the space upper bound $O(NH)$ of \Cref{theo:index} is equivalent to $O(L\lvert E \rvert)$, where $L = \max_{v \in V} \lvert\ell(v)\rvert$ is the maximum label length.
Still, the exact size of the index, strings $D_R$ and $D_F$ of \Cref{sect:index}, can be considered a compression metric on its own, and we propose the following segmentation problem.
\begin{problem}[Total edge length segmentation problem]
    Find a semi-repeat-free segmentation $S$ of \msamn minimizing the total edge label length $\sum_{(u,v) \in E} \lvert \ell(u) \rvert + \lvert \ell(v) \rvert$ of the resulting founder graph $G(S) = (V,E,\ell)$. 
\end{problem}

There are also limitations affecting the \efg framework:
\begin{itemize}
    \item The semi-repeat-free property for \msa segments does not hold if a segment contains a string $T$ that is spelled in some different part of the \msa.
    Indeed, consider string $T[1..t]$ occurring in the \msa at columns $o_1$ and $o_2$: any segment $[x..y]$ such that $o_1 \le x < y \le o_1 + t$ or $o_2 \le x < y \le o_2 + t$ cannot be semi-repeat-free and this will limit the recombination and the compression of any indexable segmentation.
    The indexing results of \Cref{sect:index} demonstrate that the edge labels $\ell(u) \ell(v)$ effectively simulate the navigation of the \efg, so a future line of research could study how methods and data structures for indexing repetitive text collections---strings $\ell(u) \ell(v)$ in this case---can be used to further compress \efgs;
    \item The segmentation of an \msa does not allow for blocks containing empty strings (\Cref{ass:empty}), limiting the power of semi-repeat-free segments in long gap runs of the \msa.
    An extension of our results in this sense---allowing empty strings and stripping them in the \efg, while still maintaining the paths corresponding to the original sequences---would naturally result in a \emph{founder directed acyclic graph}, instead of an elastic block graph, and is the subject of our ongoing research.
\end{itemize}

Our preliminary implementation of the discussed ideas can be found at \url{https://github.com/tsnorri/founder-graphs-semi-repeat-free}. We are working on a comparison with implementations of other algorithms suitable for the purpose.

\bibliographystyle{plainurl}
\bibliography{biblio}{}
\end{document}